\documentclass{article}

\usepackage{PRIMEarxiv}

\usepackage[utf8]{inputenc} 
\usepackage[T1]{fontenc}    
\usepackage{hyperref}       
\usepackage{url}            
\usepackage{booktabs}       
\usepackage{amsfonts}       
\usepackage{nicefrac}       
\usepackage{microtype}      
\usepackage{lipsum}
\usepackage{fancyhdr}       
\usepackage{subcaption}
\usepackage{amsmath}
\usepackage{amssymb}
\usepackage{amsthm}
\usepackage{graphicx}       

\graphicspath{{media/}}     

\usepackage{tikz}
\usetikzlibrary{arrows.meta, shapes.geometric, arrows, matrix}
\tikzset{%
  >={Latex[width=2mm,length=2mm]},
            base/.style = {rectangle, rounded corners, draw=black, minimum width=4cm, minimum height=1cm, text centered, font=\sffamily},
       startstop/.style = {ellipse, draw=black, minimum width=2cm, fill=gray!30, font=\sffamily},
       decision/.style = {diamond, rounded corners, minimum width=1cm, text centered, draw=black, font=\sffamily}
}

\usepackage{algorithm, algpseudocode}
\usepackage{verbatim} 

\newtheorem{theorem}{Theorem}[section]

\newtheorem{lemma}[theorem]{Lemma}
\newtheorem{assumption}{Assumption}

\DeclareMathOperator*{\argmin}{arg\,min}

\title{Wasserstein Distributional Learning}

\author{%
  Chengliang Tang \\
  Columbia University\\
  New York, NY 10027 \\
  \texttt{ct2747@columbia.edu} \\
   \And
   Nathan Lenssen \\
   Columbia University \\
   New York, NY 10027 \\
   \texttt{njl2134@columbia.edu} \\
   \And
   Ying Wei \\
   Columbia University \\
   New York, NY 10027 \\
   \texttt{yw2148@cumc.columbia.edu} \\
   \And
   Tian Zheng \\
   Department of Statistics\\
   Columbia University\\
   New York, NY 10027 \\
   \texttt{tz33@columbia.edu} \\
}

\begin{document}
\maketitle

\begin{abstract}
Learning conditional densities and identifying factors that influence the entire distribution are vital tasks in data-driven applications. Conventional approaches work mostly with summary statistics, and are hence inadequate for a comprehensive investigation. Recently, there have been developments on functional regression methods to model density curves as functional outcomes. A major challenge for developing such models lies in the inherent constraint of non-negativity and unit integral for the functional space of density outcomes. To overcome this fundamental issue, we propose Wasserstein Distributional Learning (WDL), a flexible density-on-scalar regression modeling framework that starts with the Wasserstein distance $W_2$ as a proper metric for the space of density outcomes. We then introduce a heterogeneous and flexible class of Semi-parametric Conditional Gaussian Mixture Models (SCGMM) as the model class $\mathfrak{F} \otimes \mathcal{T}$. The resulting metric space $(\mathfrak{F} \otimes \mathcal{T}, W_2)$ satisfies the required constraints and offers a dense and closed functional subspace. For fitting the proposed model, we further develop an efficient algorithm based on Majorization-Minimization optimization with boosted trees. Compared with methods in the previous literature, WDL better characterizes and uncovers the nonlinear dependence of the conditional densities, and their derived summary statistics. We demonstrate the effectiveness of the WDL framework through simulations and real-world applications.
\end{abstract}

\vfill

\section{Introduction}
\label{seq:intro}

In many scientific fields, such as economics, biology, and climate science, examining the underlying drivers of distributional heterogeneity is a powerful way for knowledge discoveries. For example, studies suggest that climate change has profoundly impacted multiple aspects of a climate outcome's distribution, including its mean, overall variability, and the frequency of extreme values \cite{field2012managing, reich2012spatiotemporal}. Figure~\ref{fig:temp-dist} displays annual distributions of daily temperature anomalies from 1880 to 2012. These distributions exhibit a shift in the mean and a substantial increase in tail behavior heterogeneity. Factor associations and their effects continue to be active areas of research \cite{fahey2017physical, lewis2017evolution}. 
Modeling of conditional distributions/densities is one way to answer such questions. Compared to traditional models that focus on conditional mean or other summary statistics, modeling conditional distributions provides a more comprehensive and unified approach to learning the complex association of interest between an outcome variable and a set of covariates.

Existing approaches modeling conditional densities are mostly likelihood-based and thus require raw individual observations. In practice, access to large-scale raw data may be limited for various reasons, including privacy protection or data storage limitations. For instance, many studies involving human subjects have strict data access due to privacy concerns, publishing only summary statistics, such as histograms, densities and quantiles, for public use. 
Income is another type of sensitive data. Usually, only their quantiles for groups under study are available. In addition,  modeling the density/quantiles is cost-effective to reduce the storage and computation burden when we handle massive data, such as internet traffic or global climate variations. In this paper, we consider (empirical) densities as \textit{functional} outcomes, and develop a Wasserstein distributional learning framework to study their associations with a large number of covariates. 

Functional regressions \cite{wang2016functional} have been applied to model probability density functions (PDFs). \cite{kneip2001inference} apply the functional data analysis (FDA) to density functions and link covariates with the principal components. Because the algorithm ignores the necessary constraints of density functions (i.e., non-negative, Borel measurable, and integrate to one), it can lead to problematic results \cite{delicado2011dimensionality}. To satisfy these inherent constraints, some of recent literature \cite{arata2017functional, boogaart2010bayes, egozcue2006hilbert, talska2018compositional, van2014bayes} adopt centered log-ratio (CLR) transformations to map PDFs onto zero-integral elements of the space of square-integrable real functions. 
A few others consider quantile-based log transformations to render functional regressions of density functions \cite{han2019additive, petersen2016functional}. 
While addressing the constraints of the functional space for density functions, such transformations lead to difficulties with model assessment and interpretation. Furthermore, when the dimension of the covariates grows, these transformation-based approaches become over-restrictive and lead to substantial bias. We illustrate the above challenges using numerical experiments in Section 3.

In this paper, we introduce Wasserstein distributional learning (WDL), a new function-on-scalar regression framework that is (1) defined on a functional space of densities, and (2) uses a Wasserstein loss function to guide the learning of conditional distributions. We propose Semi-parametric Conditional Gaussian Mixture Models (SCGMM), a sufficiently large and flexible semi-parametric conditional distribution family to account for the heterogeneity in the functional output space, and measure the functional discrepancy using the Wasserstein distance.

The Wasserstein distance \cite{villani2008optimal} measures the aggregated discrepancies between two distributions. It offers excellent convergence properties in the distributional function space \cite{villani2008optimal} and has gained popularity for its intuitive interpretation as the optimal transport costs \cite{panaretos2019statistical, villani2003topics}, in addition to its utility in real-world applications \cite{arjovsky2017wasserstein, sgouropoulos2015matching}. Compared with the other two commonly used distributional loss functions, the Kullback–Leibler divergence (KL divergence) and $L^2$ distance, the Wasserstein distance does not require the distributions to have a common support or rely on a specific transformation. It is, hence, more suitable for modeling highly heterogeneous distributions, and enjoys a more straightforward interpretation. Applying the Wasserstein distance to model conditional densities, \cite{petersen2021wasserstein, petersen2019frechet} recently introduced Fr{\'e}chet regression. One challenge with the Fr{\'e}chet regression is that it incorporates extra constraints to ensure non-crossing quantile functions, and thus becomes computationally challenging when the number of covariates is large. 



\begin{figure}[t]
	\begin{center}
		\includegraphics[width=0.6\columnwidth]{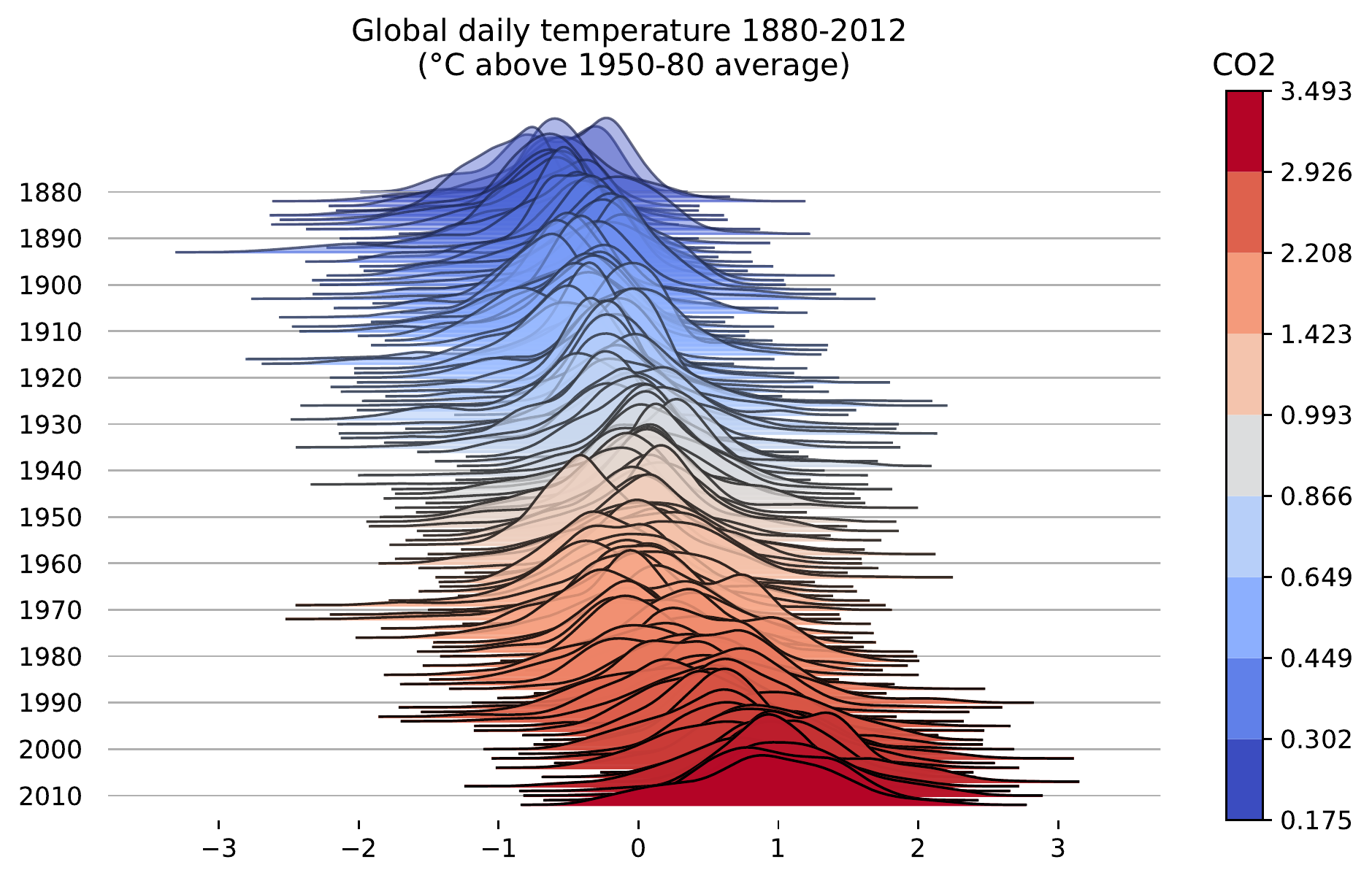}
		\caption{Historical annual distributions of daily land-surface average temperature. Temperatures are in Celsius and reported as anomalies relative to the Jan.\ 1951 -- Dec.\ 1980 average. The color of each density curve represents the corresponding annual radiative effect of increased atmospheric CO2.}
		\label{fig:temp-dist}
	\end{center}
\end{figure}

A natural choice to overcome the above challenge is to consider a family of Semi-parametric Conditional Gaussian Mixture Models (SCGMM). In this paper, we explore SCGMM with the Wasserstein loss both from a theoretical and a computational perspective. We show that under the Wasserstein distance loss, the SCGMM family is dense in the sense that it can approximate any regular functional dependence at any given precision. The proposed combination of SCGMM and Wasserstein loss is therefore a well motivated framework to model conditional distributions in a high-dimensional heterogeneous covariate space. The regression modeling of mixture components further provides interpretable results that reveal complex dependence structures and lead to insightful discoveries. Optimizing the Wasserstein loss over SCGMM model is computationally challenging due to the lack of a closed-form quantile function and the evaluation of its integral for the Wasserstein loss  \cite{kolouri2018sliced}.  This computational difficulty leaves modeling using SCGMM with the Wasserstein loss largely unexplored. In our framework, we derive an efficient iterative optimization algorithm for model training. We leverage the Majorization-Minimization (MM) algorithm \cite{lange2016mm} that converges faster than traditional gradient-based algorithms, and apply a boosting algorithm to efficiently construct nonparametric SCGMM parameter functions over covariates. 


The main contribution of our paper is a novel and efficient function-on-scalar regression framework for modeling distributional outputs. By definition, our framework satisfies the inherent constraints of density functions, and is capable of modeling highly heterogeneous outputs. We offer theoretical guarantees for the convergence of the proposed algorithm in Section 2. Compared with other methods in the literature, our proposed WDL framework better captures the nonlinear dependence of the density functions over the covariates. Moreover, this framework produces more convenient and accurate predictions for derived density summary statistics of interest. We demonstrate these advantages, in Section 3, using both simulation experiments and applications to real-world data. We summarize our framework and discuss future directions in Section 4.

\section{Wasserstein Distributional Learning}
In this section, we introduce the proposed Wasserstein distributional learning framework, illustrate its theoretical validations, and offer an efficient computational algorithm based on majorization. 

\subsection{Preliminaries}
\paragraph{Setup and Notations.} For a given subspace $\mathcal{O} \subseteq \mathbb{R}^p$, we define $\mathcal{P}(\mathcal{O})$ as the set of all Borel probability measures on $\mathcal{O}$, and $\mathcal{P}_r(\mathcal{O}) \subset \mathcal{P}(\mathcal{O})$ as its subset with a finite $r$-th moment. In this paper, we consider a joint random process $(\mathbf{X}, \mathcal{G})$, where $\mathbf{X}$ is a $p$-dimensional random covariate vector with a support $\mathcal{X} \subset \mathbb{R}^p$ and a probability density function(PDF) $f_\mathbf{X} \in \mathcal{P}(\mathcal{X})$, and $\mathcal{G}$ is a distributional response representing the distribution for an outcome $Y \in \mathbb{R}$. Let $\mathcal{G}$ be a random draw from $\mathcal{P}(\mathbb{R})$, and $\mathcal{F}$ be the joint distribution of the random process $(\mathbf{X}, \mathcal{G})$ on the product space $\mathcal{X} \times \mathcal{P}(\mathbb{R})$.  Assume a random sample of i.i.d.\ draws from $\mathcal{F}$: $\mathcal{D}_{\mathcal{F}} = \{(x_i, g_i)\}_{i=1}^n$. Taking the climate data in the Introduction Section as an example, for the $i$-th year, $g_i$ is the observed empirical distribution of daily temperatures,  while $x_i$ consists of human activities, sun activities, and other covariates of that year.  Our goal is to model the expected conditional distribution $\mathbb{E}_\mathcal{F}(\mathcal{G}|\mathbf{X})$ from the random sample $\mathcal{D}_{\mathcal{F}} = \{(x_i, g_i)\}_{i=1}^n$. 


\paragraph{Wasserstein Distance.} 
Suppose $f_1$ and $f_2 $ are two continuous PDFs from $\mathcal{P}_2(\mathbb{R})$, the $2$-Wasserstein distance $W_2(f_1, f_2)$ between $f_1$ and $f_2$ is defined as
\begin{equation*}
W_2(f_1, f_2) = \Big[\int_0^1 \big(F_1^{-1}(s) - F_2^{-1}(s)\big)^2 ds\Big]^{\frac{1}{2}},
\end{equation*}
where $F_1^{-1}$ and $F_2^{-1}$ are the quantile functions derived from $f_1$ and $f_2$, respectively. As mentioned in Introduction, Wasserstein distance is a well-defined metric in the space of distribution functions, and thus is particularly suitable for modeling heterogeneous functional outcomes. In this paper, we focus on the one-dimensional $2$-Wasserstein distance. For more general definitions of the Wasserstein distance in multi-dimensional spaces, we refer the readers to \cite{panaretos2019statistical, villani2008optimal}.  

\subsection{Overview of Wasserstein Distributional Learning}
Functional mapping from the scalar covariates to the density outputs is challenging due to the infinite dimensions of the output space. To circumvent this difficulty, we assume a semi-parametric conditional distribution family, 
\begin{equation}
    \mathfrak{F} \otimes \mathcal{T}  = \{f_{\boldsymbol{\theta}}\circ  \boldsymbol{\tau}(x) \mid  f_{\boldsymbol{\theta}} \in \mathfrak{F},\,  \boldsymbol{\theta}\in \boldsymbol{\Theta} \subset \mathbb{R}^q; \; 
    \boldsymbol{\tau}(\cdot) \in \mathcal{T}(  \mathbf{\mathcal X}, \boldsymbol{\Theta}); \; x\in \mathcal{X}\subset \mathbb{R}^p\},
\label{eq:semi-para}
\end{equation}
where $\mathfrak{F} = \{f_{\boldsymbol{\theta}} | \boldsymbol{\theta} \in \boldsymbol{\Theta} \subset \mathbb{R}^q\}$ is a parametric distribution family, $\mathcal{T}(\mathcal{X}, \boldsymbol{\Theta})$ is a non-parametric functional family of mappings from the covariate space $\mathcal{X}$ to the distribution parameter space $\boldsymbol{\Theta}$. 

The above semi-parametric conditional distribution family is sufficiently large and flexible such that the expected conditional distributions $\mathbb{E}_\mathcal{F}(\mathcal{G}|\mathbf{X} = x)$ can be well approximated by its elements. In other words, there exists a mapping $\boldsymbol{\tau}(\cdot)  \in \mathcal{T}(\mathcal{X}, \boldsymbol{\Theta})$, such that for any $x \in \mathcal{X}$, $\mathbb{E}_\mathcal{F}(\mathcal{G}|\mathbf{X} = x) \approx f_{\boldsymbol{\theta} = \boldsymbol{\tau}(x)}$. 
For a given set of observations $\mathcal{D}_{\mathcal{F}} = \{(x_i, g_i)\}_{i=1}^n$, we propose the following functional Wasserstein regression to identify the optimal mapping $\widehat{\boldsymbol{\tau}}(\cdot)$ in a specified functional space $\mathcal{T}(\mathcal{X}, \boldsymbol{\Theta})$ and minimize the Wasserstein loss that is evaluated at the observed empirical distributions and covariates:
\begin{equation} 
\widehat{\boldsymbol{\tau}}(\cdot) = \argmin_{\boldsymbol{\tau}(\cdot) \in \mathcal{T}(\mathcal{X}, \boldsymbol{\Theta})} \sum_{i=1}^n W_2^2(g_i, f_{\boldsymbol{\theta} = \boldsymbol{\tau}(x_i)}) = \argmin_{\boldsymbol{\tau}(\cdot) \in \mathcal{T}(\mathcal{X}, \boldsymbol{\Theta})}  \sum_{i=1}^n
\int_0^1
\big(F_{g_i}^{-1}(s) - F_{\boldsymbol{\theta} = \boldsymbol{\tau}(x_i)}^{-1}(s)\big)^2 ds,   
\label{eq:1}
\end{equation}
where $F_{g_i}^{-1}(s) $ is the quantile function derived from $g_i$ and $F_{\boldsymbol{\theta} = \boldsymbol{\tau}(x_i)}^{-1}(s)$ is that derived from $f_{\boldsymbol{\theta} = \boldsymbol{\tau}(x_i)}$. For simplicity, we refer to them as $F_{\boldsymbol{\tau}(x_i)}^{-1}$ and $f_{\boldsymbol{\tau}(x_i)}$ in the rest of the paper.

In the sections below, we consider a class of Semi-parametric Conditional Gaussian Mixture Model (SCGMM) as $\mathfrak{F} \otimes \mathcal{T}$, establish its universal approximation property under the Wasserstein loss, and prove the uniform consistency of the resulting estimates $\widehat{\boldsymbol{\tau}}(\cdot)$.  We also present a Majorization-minimization boosting framework for estimating $\widehat{\boldsymbol{\tau}}(\cdot)$ efficiently.

\subsection{Semi-parametric Conditional Gaussian Mixture Model}
\label{sec:tree}
To account for the heterogeneity in the functional output, we propose a class of Semi-parametric Conditional Gaussian Mixture Model (SCGMM), where the mean, variance and weight parameters of the Gaussian Mixture are unspecified functions of the covariates $\mathcal{X}$, which can be written as
\begin{equation}
f_{\boldsymbol{\tau}(x)} = \sum_{k=1}^K \pi_k(x) \mbox{N}\{\mu_k(x), \sigma^2_k(x)\},
\label{eq:model}
\end{equation}
where $\mbox{N}$ represents a Gaussian distribution, $\mu_k(x)$ and $\sigma^2_k(x)$ are the mean and variance of the $k$-th component, and $\pi_k(x)$ is the weight associated with the $k$-th component. 
We assume that all the parameters are unknown functions of the covariate $x$, and denote $\boldsymbol{\tau}(x) = \{\pi_k(x), \mu_k(x),\sigma^2_k(x)\}_{k=1}^K$ as the collection of all \emph{parameter functions} of model~\eqref{eq:model}.

By definition, the SCGMM model automatically satisfies the non-negativity and unit-integral constraint of density functions. We propose to use 2-Wasserstein distance for the functional regression of density/distribution outcomes. The validity of the proposed Wasserstein distributional learning is guaranteed by the two theorems that we established in the subsequent section. In this first theorem, we show that the SCGMM is dense in the Wasserstein space so that it well approximates any regular conditional distributions; In the second one, we prove the optimizer from the Wasserstein regression, $f_{\widehat{\boldsymbol{\tau}}(x)}$, is uniformly consistent to $\mathbb{E}_\mathcal{F}(\mathcal{G}|\mathbf{X})$, the true conditional distribution over $\mathcal{X}$. In practice, the model training of SCGMM under the Wasserstein distance loss is prohibitively difficult, which is the major obstacle limiting its applicability. The Wasserstein loss function is non-convex over the mixture model parameters $\{\pi, \mu, \sigma\}$, and the computation of the loss gradient is extremely noisy in practice. In our Wasserstein distributional learning framework, we designed an efficient boosted Majorization-Minimization optimization procedure to address the computational challenges, and also established its convergence property. 



\subsection{Universal Approximation and Consistency}
In this section, we investigate the universal approximation and consistency of the SCGMM under the Wasserstein distance. We first establish the dense property of the Gaussian Mixture Model (GMM) under the Wasserstein distance in the following Lemma, then generalize the results to the SCGMM in the later theorem.

\begin{lemma}
Let $\mathfrak{F}_G \subset \mathcal{P}_2(\mathbb{R})$ be the family of finite Gaussian mixture distributions over the real line $\mathbb{R}$, i.e., 
\begin{equation*}
    \mathfrak{F}_G = \Bigg\{\sum_{k = 1}^K \pi_k \mbox{N}(\mu_k, \sigma^2_k)\Bigg|K \in \mathbb{N}_+  \Bigg\}.
\end{equation*}
Then the family $\mathfrak{F}_G$ is dense in $\big(\mathcal{P}_2(\mathbb{R}), W_2\big)$, i.e., the set of probability measures of a finite second moment with the 2-Wasserstein distance.
\label{lemma:dense}
\end{lemma}

\begin{proof}
In the Appendix.
\end{proof}

Lemma \ref{lemma:dense} is a Wasserstein generalization of Theorem 1 in \cite{park1991universal}, which proves the dense property of the Gaussian Mixture Model under the $L^p$ distance for $\forall p \geq 1$. It indicates that a large class of density functions can be well approximated by a Gaussian Mixture distribution with a sufficiently large number $K$ and appropriately chosen component-wise parameters. 

Next, we show that such universal approximation property also holds for the SCGMM under the Wasserstein distance, where all the parameters are step functions of $\mathbf{X}$. Following our notations, we denote the distribution function of $\mathbf{X}$ by $P_\mathbf{X}$. For simplicity, we denote $H(x) \triangleq \mathbb{E}_{\mathcal{F}}(\mathcal{G}|\mathbf X = x) \in \mathcal{P}(\mathbb{R})$  as the expected conditional density of $Y$ given $\mathbf X = x$, and denote both $\boldsymbol{\tau}$ and $\boldsymbol{\tau}(\cdot)$ as the mapping from $\mathcal{X}$ to $\boldsymbol{\Theta}$ without distinction. 

Now, we introduce the following sufficient assumptions for the universal approximation property of SCGMM. 

\begin{assumption}
1). The covariate $\mathbf{X}$ follows a light-tailed distribution, i.e., there exist constants $\lambda$ and $M_0$, such that $P_\mathbf{X}(\|\mathbf{X}\|_2 > M) < \exp(-\lambda M)$ for any $M > M_0$;

2). $H(x)$ has a finite second moment for $x \in \mathcal{X}$ almost everywhere;

3). $H(x): \mathcal{X} \longrightarrow \mathcal{P}(\mathbb{R})$ is Lipschitz continuous, i.e., there exists a real constant $L > 0$ such that, for all $x_1$ and $x_2$ in $\mathcal{X}$,  
\begin{equation*}
W_2(H(x_1), H(x_2)) \leq L \|x_1 - x_1 \|_2.
\end{equation*}

\end{assumption}

Regarding the above assumptions, the first one is to constrain the speed of decay of the covariates, the second one is to guarantee that the Wasserstein distance is finite for the expected conditional density $H(x)$, and the third one adds continuity constraint to $H(x)$ such that its functional dependence over the covariates $x$ can be approximated by simple functions. Under them, we have the main theorem on the uniform approximation of SCGMM under Wasserstein distance.

\begin{theorem}
For any $\varepsilon > 0$, there exists a positive integer $K > 0$, and corresponding Gaussian mixture regression 
$f_{\boldsymbol{\tau}(x)} = \sum_{k=1}^K \pi_k(x) \mbox{N}\left(\mu_k(x), \sigma^2_k(x)\right)$,
such that
\begin{equation*}
    \int_{x\in \mathcal{X}} W_2(H(x), f_{\boldsymbol{\tau}(x)})d P_\mathbf{X}(x) < \varepsilon,
\end{equation*}
where $\boldsymbol{\tau}(x) = \{\pi_k(x), \mu_k(x), \sigma^2_k(x)\}$ are all scalar-valued step functions of $x$. 
\label{theorem:approximation}
\end{theorem}

\begin{proof}
In the Appendix.
\end{proof}

To our best knowledge, it is the first attempt to establish the uniform universal approximation of SCGMM under Wasserstein distance. Next we establish the consistency of the estimated SCGMM from Wasserstein regression \eqref{eq:1}. We denote $\mathcal{M}_n(\boldsymbol{\tau}) = \frac{1}{n}\sum_{i = 1}^n W^2_2(f_{\boldsymbol{\tau}(x_i)}, g_i)$ and $\mathcal{M}(\boldsymbol{\tau}) = \int_{\mathcal{X} \times \mathcal{P}(\mathbb{R})}W^2_2(f_{\boldsymbol{\tau}(x)}, g) d\mathcal{F}(x, g)$. Here, $\mathcal{M}_n, \mathcal{M}: \mathcal{T}(\mathcal{X}, \boldsymbol{\Theta}) \rightarrow \mathbb{R}_+$ are non-negative functions defined over the functional space $\mathcal{T}(\mathcal{X}, \boldsymbol{\Theta})$. Further, we introduce metric $d(\cdot, \cdot)$ over $\mathcal{T}(\mathcal{X}, \boldsymbol{\Theta}) \times \mathcal{T}(\mathcal{X}, \boldsymbol{\Theta})$ by $d(\boldsymbol{\tau}_1, \boldsymbol{\tau}_2) \triangleq \sup_{x \in \mathcal{X}} \|\boldsymbol{\tau_1}(x) - \boldsymbol{\tau_2}(x) \|$. With these notations, we introduce the following additional sufficient assumptions for the consistency property of SCGMM.

\begin{assumption}
1). The map $\boldsymbol{\theta} \mapsto f_{\boldsymbol{\theta}}$ is continuous in the sense that for any sequence $\{\boldsymbol{\theta}_n\} \subset \boldsymbol{\Theta}$ and point $\boldsymbol{\theta}_0 \in \boldsymbol{\Theta}$,  $\|\boldsymbol{\theta}_n - \boldsymbol{\theta}_0 \| \rightarrow 0$ implies $W_2(f_{\boldsymbol{\theta}_n}, f_{\boldsymbol{\theta}_0}) \rightarrow 0$;

2). For any $g \in \mathcal{P}_2(\mathbb{R})$, the minimizer set $\argmin_{\boldsymbol{\theta} \in \boldsymbol{\Theta}}W_2(g, f_{\boldsymbol{\theta}})$ is non-empty and only has one element.
\end{assumption}

Regarding the above assumptions, the first one is about the continuity of the model-defined distribution $f_{\boldsymbol{\theta}}$ over the parameter $\boldsymbol{\theta}$, and the second one is to ensure the uniqueness of the minimizer. Under them, we have the main theorem on the consistency of SCGMM with the Wasserstein distance.

\begin{theorem}
Under Assumptions 1-2, suppose $\boldsymbol{\tau}_0 = \argmin_{\boldsymbol{\tau} \in \mathcal{T}(\mathcal{X}, \boldsymbol{\Theta})}\mathcal{M}(\boldsymbol{\tau})$, and the SCGMM estimators $\widehat{\boldsymbol{\tau}}_n = \argmin_{\boldsymbol{\tau} \in \mathcal{T}(\mathcal{X}, \boldsymbol{\Theta})}\mathcal{M}_n(\boldsymbol{\tau})$ all lie in a compact set $S \subset \mathcal{T}(\mathcal{X}, \boldsymbol{\Theta})$, then $\widehat{\boldsymbol{\tau}}_n$ are consistent in the sense that for every $\varepsilon > 0$,
\begin{equation*}
    \mathbb{P}(d(\widehat{\boldsymbol{\tau}}_n(x), \boldsymbol{\tau}_0(x)) \geq \varepsilon) \rightarrow 0.
\end{equation*}

\label{thm:consistency}
\end{theorem}
\begin{proof}
In the Appendix.
\end{proof}

\subsection{Majorization-Minimization Optimization}
\label{sec:finite}

Optimizing over the Wasserstein distances has been known to be computationally challenging \cite{bernton2019parameter, kolouri2017optimal}, which largely limits its applications. In this section, we introduce a novel and efficient EM-like optimization algorithm that makes Wasserstein distributional learning computationally feasible.  For easier illustration, we introduce the concept in a more simplistic setting without covariate.  We will expand to the full algorithm, integrating boosting machines, in Section~\ref{sec_boosting}.

Suppose $\mathcal{D}_{\mathcal{Y}} = \{y_i\}_{i = 1}^n$ is an i.i.d. random sample, where the empirical distribution is $F_{\mathcal{D}}$, and $\{F_{\boldsymbol{\theta}}: \boldsymbol{\theta} \in \boldsymbol{\Theta}\}$ is a target parametric distribution family. We aim to find the optimal parameters $\widehat{\boldsymbol{\theta}}$ that minimize $L(\boldsymbol{\theta}) = W^2_2(F_{\mathcal{D}}, F_{\boldsymbol{\theta}})$ over all the possible $\boldsymbol{\theta} \in \boldsymbol{\Theta}$. 

To understand the computation challenges in Wasserstein distributional learning, we begin with a simple case. We consider just one distribution component and the targeting distribution 
$F_{\boldsymbol{\theta}}$ belongs to a location-scale family, $\mathfrak{F} =\{ F_{\boldsymbol{\theta}}(y) = F_0(\frac{y - \mu}{\sigma}): \boldsymbol{\theta} = (\mu, \sigma), \mu\in R, \sigma\in R^+\}$, expanding from a standard distribution $F_0(Y)$, then the Wasserstein distance  can be formulated as
\begin{equation*}
L(\boldsymbol{\theta}) = W^2_2(F_{\mathcal{D}}, F_{\boldsymbol{\theta}}) = \int_0^1 \big(F^{-1}_{\boldsymbol{\theta}}(s) - F_{\mathcal{D}}^{-1}(s)\big)^2 ds
= \int_0^1 \big(\mu + \sigma \cdot F^{-1}_0(s) - F_{\mathcal{D}}^{-1}(s)\big)^2 ds.
\end{equation*}
The major convenience yielded by the location scale family is that their quantile functions can be expressed in terms of the linear function of location and scale parameters. Therefore, the above formulation can be treated as a simple OLS regression in the quantile space, with $\mu$ and $\sigma$ representing the parameters of intercept and slope, respectively. Specifically, the minimum Wasserstein estimations of $(\mu, \sigma)$ are 
\begin{equation}
\widehat{\mu} = \int F_{\mathcal{D}}^{-1} ds - \widehat{\sigma} \cdot \int F_0^{-1} ds, \quad
\widehat{\sigma} = \frac{\int F_0^{-1} F_{\mathcal{D}}^{-1} ds - \int F_0^{-1} ds \cdot \int F_{\mathcal{D}}^{-1} ds}{\int (F_0^{-1})^2 ds - (\int F_0^{-1} ds)^2}.
\label{eq:normal}
\end{equation}
In practice, the integrals can be approximated numerically by summation over a sequence of discrete quantile levels $0 < q_1 < ... < q_M < 1$ or using Monte Carlo integration. As a special case of location-scale distribution, the estimations \eqref{eq:normal} also hold for  Gaussian distributions.

When the target distribution $F_{\boldsymbol{\theta}}$ is a Gaussian mixture, the optimization over $L(\boldsymbol{\theta})$  is much more challenging. The major difficulty stems from the complicated nonlinear dependence of the quantile functions $F^{-1}_{\boldsymbol{\theta}}(s)$ over the model parameters $\boldsymbol{\theta} = \{(\pi_k, \mu_k, \sigma_k)\}_{k=1}^K$. Moreover, the Wasserstein loss function does not satisfy the desired convex property, which greatly hinders the stability and convergence of gradient-based optimization algorithms. To address these challenges, we utilize the Majorization-Minimization (MM) procedure and identify an alternative loss function that will serve as the upper bound of the Wasserstein loss. In this case, the main objective is to simultaneously update the decomposition of the empirical distribution $f_\mathcal{D}$ and model distribution $f_{\boldsymbol{\theta}}$. 

Additionally, to ensure the model identifiability of the parametric continuous distribution $f_{\boldsymbol{\theta}}$, we add an order constraint to the component expectations throughout the paper, stating
$\mu_1 \leq \mu_2 \leq ... \leq \mu_K$, where $\mu_k = E_{X \sim f_k}(X), k = 1, ..., K$. Meanwhile, it should be noted that when we refer to an empirical distribution $g = f_{\mathcal{D}}$, the above constraint is not necessary since no model parameters are required. With the definition, we introduce the following theorem, which lays the foundation of the proposed MM algorithm.

\begin{theorem}
	For any two continuous PDFs $f \in \mathcal{P}(\mathbb{R})$ and $g \in \mathcal{P}(\mathbb{R})$, along with any mixture decomposition of $f = \sum_{k=1}^K \pi_k \cdot f_k$ and $g = \sum_{k=1}^K \pi_k \cdot g_k$, the following inequality holds
	\begin{equation}
	W_2^2(f, g)  \leq  \sum_{k=1}^K \pi_k W_2^2(g_k, f_k),
	\label{eq:surrograte}
	\end{equation}
	and the equality holds when
	\begin{equation}
	g_k(x) = g(x)\cdot \frac{f_k \circ F^{-1} \circ G(x)}{\sum_{l=1}^K \pi_l f_l \circ F^{-1} \circ G(x)}, \ \forall x \in \mathbb{R}, \quad \text{for}\quad k = 1, ..., K.
	\label{eq:g-k}
	\end{equation}
	Here, $F^{-1}$ is the QF of $f$, and $G$ is the CDF of $g$.
	\label{thm:ineq}
\end{theorem}
\begin{proof}
In the Appendix.
\end{proof}

The above theorem shows that for any two continuous density distributions $f, g$ and their mixture decompositions, as long as they share the same mixture component weights, their Wasserstein distance can be bounded by the weighted sum of the Wasserstein distances between each pair of mixture components. By plugging in $g = f_\mathcal{D}, f = f_{\boldsymbol{\theta}}$, and assuming $f_{\boldsymbol{\theta}}$ belongs to the family of Gaussian mixture distributions with $\boldsymbol{\theta} = \{(\pi_k, \mu_k, \sigma_k)\}_{k=1}^K$, then the left side of Equation \eqref{eq:surrograte} in Theorem \ref{thm:ineq} becomes the target loss function $L(\boldsymbol{\theta}) = W_2^2(f_\mathcal{D}, f_{\boldsymbol{\theta}})$. Moreover, if we treat the mixture decomposition component $\{g_k\}_{k=1}^K$ of the empirical distribution as the latent parameters, the right side of Equation \eqref{eq:surrograte} provides a natural upper bound of $L(\boldsymbol{\theta})$. This can serve as a surrogate function that \textit{majorizes} the original objective function. Denoting the right side of Equation \eqref{eq:surrograte} by $R(\boldsymbol{\nu}, \boldsymbol{\theta}) = \sum_{k=1}^K\pi_kW_2^2(g_k, f_k)$ with $\boldsymbol{\nu} = \{g_k\}_{k=1}^K$, $\boldsymbol{\theta} = \{(\pi_k, \mu_k, \sigma_k)\}_{k=1}^K$ and $f_k = \mbox{N}(\mu_k, \sigma_k^2)$, we have $L(\boldsymbol{\theta}) \leq R(\boldsymbol{\nu}, \boldsymbol{\theta})$ by Theorem \ref{thm:ineq}. 

By utilizing the Majorization-Minimization (MM) procedure, we propose an iterative optimization algorithm to find the Minimum Wasserstein estimation of GMM, letting $m$ be the iteration indicator. We denote  $g_k^{(m)}$, $k=1,...,K$,  as the mixture decomposition of $f_\mathcal{D}$ in the $m$-th iteration,  $f_k^{(m)} = \mbox{N}(\mu_k^{(m)}, \sigma_k^{(m)})$ as those of $f_{\boldsymbol{\theta}}$, and $\pi^{(m)}$ as the component weights. The $m$-th iteration consists of the following three important sub-steps.

1. {\bf [$f_k^{(m-1)} \longmapsto f_k^{(m)}$]} Given $\pi_k^{(m-1)}$ and $g_k^{(m-1)}$, this step finds the optimal $f_k^{(m)} = \mbox{N}(\mu_k^{(m)}, \sigma_k^{(m)})$ for minimizing $\sum_{k=1}^K \pi_k^{(m-1)} W_2^2(f_k, g_k^{(m-1)})$. In order to minimize the sum of Wasserstein distances, we can simply find the optimal $f_k^{(m)}$ for each term $W_2^2(f_k, g_k^{(m-1)})$. When $f_k$ is from the Gaussian family, this step corresponds to Equation \eqref{eq:normal} .
	
2. {\bf [$\pi_k^{(m-1)} \longmapsto \pi_k^{(m)}$]} Given $f_k^{(m)}$, this step finds the optimal $\pi_k^{(m)}$ for minimizing the target loss $L(\boldsymbol{\theta}) = W_2^2(f_{\boldsymbol{\theta}}, f_\mathcal{D})$ with $f_{\boldsymbol{\theta}} = \sum_{k=1}^K \pi_k f_k^{(m)}$. This step is the most challenging part in the optimization as there is no explicit formula for the optimal $\pi_k$. Here, we provide two alternative solutions.
	
The first solution is based on gradient descent. The derivative of loss function $L(\boldsymbol{\theta})$ with respect to $\pi_k$ is
\begin{equation*}
\frac{\partial L(\boldsymbol{\theta})}{\partial \pi_k} = \int_{\mathbb{R}} (G^{-1} \circ F_{\boldsymbol{\theta}}(t) - t) \cdot F_k(t) dt.
\end{equation*}
In this case, $F_k$ is the CDF of $f_k^{(m)}$. The convergence of this solution is guaranteed by the convexity of $L(\boldsymbol{\theta})$ over $\pi$, according to Theorem \ref{thm:ineq}. The second solution is to use the Maximization step in EM algorithm for approximating the optimal $\pi$, which is expressed as 
\begin{equation*}
\pi_k^{(m)} = \int_{\mathbb{R}} g^{(m-1)}_k(x) \cdot \frac{\pi_k^{(m-1)}f_k^{(m)}(x)}{\sum_{j=1}^K\pi_j^{(m-1)}f_j^{(m)}(x)} dx.
\end{equation*} 
Since the Wasserstein distance characterizes the weak topology in the density distribution space, it is dominated by KL divergence (strong topology) in the limiting case. This solution does not find the optimal $\pi$ under the Wasserstein loss, but works well in practice and is much easier to implement.
	
3. {\bf [$g_k^{(m-1)} \ \longmapsto \ g_k^{(m)}$]} Given $\pi_k^{(m)}$ and $f_k^{(m)}$, this step finds the optimal $g_k^{(m)}$ for minimizing $\sum_{k=1}^K \pi_k^{(m)} W_2^2(f_k^{(m)}, g_k)$. In other words, the goal is to find the optimal mixture decomposition for the empirical distribution $g = f_\mathcal{D} = \sum_{k=1}^K \pi_k^{(m)} g_k$, such that the surrogate loss $R(\boldsymbol{\nu}, \boldsymbol{\theta})$ is minimized. Since the surrogate function has a lower bound as the original loss function, and the equality conditions are given in Theorem \ref{thm:ineq}, we only need to calculate the updated decomposition as in Equation \eqref{eq:g-k}. Note that this step is quite similar with the E-step in the conventional EM algorithm, and the only difference is that $x$ is replaced by $F_{\boldsymbol{\theta}}^{-1} \circ G(x)$.

Following the Theorem \ref{thm:ineq}, we have the following inequality chain
\begin{equation*}
\begin{split}
L(\boldsymbol{\theta}^{(m)}) &= L(\pi^{(m)}, \mu^{(m)}, \sigma^{(m)})  \leq L(\pi^{(m-1)}, \mu^{(m)}, \sigma^{(m)}) \leq R(\{g_k\}^{(m-1)}, \pi^{(m-1)}, \mu^{(m)}, \sigma^{(m)}) \\
&\leq R(\{g_k\}^{(m-1)}, \pi^{(m-1)}, \mu^{(m-1)}, \sigma^{(m-1)}) = L(\pi^{(m-1)}, \mu^{(m-1)}, \sigma^{(m-1)}) = L(\boldsymbol{\theta}^{(m-1)}).
\end{split}
\end{equation*} It indicates that the original loss function $L(\boldsymbol{\theta})$ decreases during the optimization of the upper bound $R(\boldsymbol{\nu}, \boldsymbol{\theta}) = R\big(\{ g_k\}_{k=1}^K, \pi, \mu, \sigma \big)$ in each iteration step. The outline algorithm hence converges to the minimum of the original loss function. 

\subsection{Boosted Majorization-Minimization Optimization for Wasserstein Distributional Learning } 
\label{sec_boosting}
In this section, we generalize the minimum Wasserstein estimation to the case of Conditional Gaussian Mixtures, where the distribution parameters $\boldsymbol{\theta}$ are unknown functions of $\mathbf{X}$, i.e.   $\boldsymbol{\theta}|x = \boldsymbol{\tau}(x) = \{\pi_k(x), \mu_k(x), \sigma_k(x)\}_{k=1}^K$.  Specifically,  we model them by way of boosted trees \cite{friedman2001greedy}, which are the sum of a number of decision trees learned by boosting. 

In order to satisfy the parameter constraints introduced in Section~\ref{sec:tree}, transformations are necessary for the raw GMM parameters. For the mixing weights $\pi_k(x)$, they must satisfy the constraints $\sum_{k=1}^K \pi_k(x) = 1$ and $\pi_k(x) \geq 0 \ \forall k$. This is achieved by reparameterization through the softmax function
$\pi_k(x) = \frac{\exp(\alpha_k(x))}{\sum_{k=1}^K\exp(\alpha_k(x))}$, where $\alpha_k(x)$'s can take values in $\mathbb{R}$ and therefore modeled as the outputs of boosted trees. Similarly, the scale parameters $\sigma_k(x)$ are represented in terms of the exponential of the boosted tree outputs in order to satisfy the positive constraint $\sigma_k(x) = \exp(z_k(x)) > 0$.
Since the mean components, $\mu_k(x)$'s, are not subject to any constraints, we can simply model them as the direct outputs of boosted tree regression models.

Under the above reparameterizations, we fit the mixture regression model via an iterative boosting algorithm. With a predefined learning rate $\eta > 0$ and simple initial functions $\widehat{\boldsymbol{\tau}}^{(0)}(x) = 0$, suppose we run the algorithm for $M > 0$ steps. At each step $1 \leq m \leq M$, we calculate the current values of all parameters through $\widehat{\boldsymbol{\tau}}^{(m-1)}(x)$ and then determine the next optimal parameter values at each data point following the steps introduced in Section~\ref{sec:finite}. Afterward, a regression tree is fitted for each parameter to approximate the difference between the new value and the current stage. Finally, it is added to the tree ensembles $\widehat{\boldsymbol{\tau}}^{(m-1)}(x)$ to have the updated model $\widehat{\boldsymbol{\tau}}^{(m)}(x)$.  Besides using a fixed number of iterations, we can also use early stopping \cite{yao2007early} to avoid overfitting and accelerate iterative optimization. Specifically, we calculate the Wasserstein loss $\sum_{i = 1}^n W^2(g_i, f_{\boldsymbol{\tau}(x_i)})$ on a validation set along the training process and stop when it is no loner decreasing. To ensure identifiability, we add an extra constraint onto the component means: for any given $\mathbf{X} = x$, it is assumed that $\mu_1(x) \leq \mu_2(x) \leq ... \leq \mu_K(x)$. Therefore, in each optimization step, after updating the component parameters, the components will be sorted before feeding to the boosting machine. Further details can be found in Algorithm B.1 and Figure B.1 of Appendix B.

As opposed to the gradient-based optimization algorithm introduced in \cite{bishop1994mixture, rothfuss2019conditional}, our model training framework utilizes the additive structure of tree ensembles and the upper bound of the Wasserstein loss to achieve a more stable and efficient solution path. In actuality, any machine learning algorithms can be used to represent and estimate the nonparametric coefficient functions. Compared with other models, such as polynomial functions and neural networks, the boosted decision trees achieve a balance between the expressive power and the generalization ability. Moreover, the tree structure greatly improves the model transparency and interpretability.

\section{Experiments}
In this section, we present numerical experiments that demonstrate the estimation performance and prediction accuracy of the proposed Wasserstein distributional learning (WDL) framework using simulations and two real-world applications. In all the experiments, we compared the WDL framework with three existing density regression methods. The first one is the global Fr{\'e}chet regression \cite{petersen2019frechet}, a generalization of linear regression in the quantile functional space under the Wasserstein loss, and the second method is the B-spline smoothed density regression with a centered log-ratio (CLR) transformation \cite{talska2018compositional}, which, for reasons of simplicity, is referred to as the CLR regression for the rest of the paper. The last comparison method is the Mixture Density Network (MDN) \cite{bishop1994mixture}. Different from the previous methods, MDN generats conditional density estimation from scalar outcomes, which is a related but different task than functional regression modeling of density functions. In our experiments, we implemented MDN using Monte Carlo samples from the observed functional outcomes. The main ideas of these methods have been introduced in Section \ref{seq:intro} with more details provided in the Appendix. Reproducible codes for generating all results are available on GitHub (\url{https://github.com/ChengliangTang/Wasserstein-Distributional-Learning}).

\subsection{Simulation Study}
In this experiment, we consider multivariate covariates $\mathbf{X} = (\mathbf{x_1}, \mathbf{x_2}, \mathbf{x_3})$ that are mutually independent and follow the uniform distribution on $[-1, 1]$. The conditional density outcomes $\mathcal{G}$ are generated as follows.
\begin{equation}
\mathcal{F}(\mathcal{G} | \mathbf{X} = x) = \pi_1(x) \cdot f_1(x) + \pi_2(x) \cdot f_2(x),
\label{eq:simulation}
\end{equation}
where
\begin{equation*}
\begin{split}
\pi_1(x) = \frac{1}{1 + \exp(x_3)}, \, &\ \pi_2(x) = \frac{\exp(x_3)}{1+\exp(x_3)}, \\
f_1(x) = \mbox{N}\big(x_1 + \varepsilon, (|x_2|+0.5)^2\big),\, &\ f_2(x) = \mbox{N}\big(2x_2^2 + 2 + \varepsilon, (|x_1|+0.5)^2\big),   
\end{split}
\end{equation*}
with independent random noise variable $\varepsilon \sim \mbox{N}(0, \omega^2)$.

By design, we assume that the conditional distribution $\mathcal{F}(\mathcal{G} | \mathbf{X} = x)$ is a Gaussian mixture. The component-wise means and variances are functions of $\mathbf{x_1}$ and $\mathbf{x_2}$, while the component weights, $\pi_1(x)$ and $\pi_2(x)$, are governed by $\mathbf{x_3}$. In addition, we let $\mu_1(x) \leq \mu_2(x), \forall x \in \mathcal{X}$ to avoid identifiability issues. We also incorporate an additive random noise $\varepsilon$ to the component-wise means $\mu_1(x)$ and $\mu_2(x)$, which is independent of all $\mathbf{X}$ variables. The additive random noise follows a zero-mean Gaussian distribution $\varepsilon \sim \mbox{N}(0, \omega^2)$, which controls the noise level by the standard deviation $\omega$.

Given a noise level $\omega$, we generate, from the model specified in \eqref{eq:simulation}, $N = 200$ random samples $(x_i, g_i) \sim \mathcal{F}(\mathbf{X}, \mathcal{G})$. We first take independent random draws $(x_i, \varepsilon_i)$ from the proposed distributions, and then generate $g_i$, which is represented by the empirical quantile function of $300$ random draws from the conditional distribution $\mathcal{F}(\mathcal{G} | \mathbf{X} = x_i, \varepsilon_i)$. We apply the proposed WDL to the generated $\{(x_i, g_i)\}_{i=1}^{N=200}$ to estimate the parameter functions $\boldsymbol{\tau}(x) = \{\pi(x), \mu(x), \sigma^2(x)\}$, and derive the expected conditional distributions $\mathbb{E}(\mathcal{G} | \mathbf{X})$. 

Using simulated data, we evaluate the performance using three different measures: accuracy in estimating parameter functions, accuracy in estimating conditional distributions, and accuracy in predicting the functional outputs. The first two measures focus on the estimation performance of each method using the training set, and the third measure evaluates their generalization abilities from a training set to an independent test set. For the first two measures, we evaluate their average performance over 500 Monte Carlo replications. On each Monte Carlo replication, we begin by randomly splitting the data into the training set (80\%) and the validation set (20\%), along with choosing the best tuning parameter based on validation results. Then, we refit each model with the best tuning parameter over the entire data set. For the third measure, following what we would use in a real data scenario, we evaluate the performance using a nested five-fold cross validation. In order to minimize the optimistic bias in performance evaluation, hyper-parameter selection and model training were performed using another layer of train-valid split over the training folds, at which point we evaluated the prediction loss on the held-out test fold. We applied parameter tuning to WDL, CLR regression and MDN. Because there are no tuning parameters in the Fr{\'e}chet regression, this tuning step was unnecessary.

\paragraph{Accuracy in estimating the parameter functions.} We first investigate how well the proposed WDL estimates the component-specific parameter functions. Following the algorithms outlined in Section 2.5, we estimate the parameter functions $\pi(x)$, $\mu_1(x)$, $\mu_2(x)$, $\sigma_1(x)$ and $\sigma_2(x)$. To assess the model fit against the ground truth, we use partial dependence plots (PDP) with respect to $\mathbf{x_1}, \mathbf{x_2}$ and $\mathbf{x_3}$ \cite{friedman2001greedy} and compare them with marginal dependence functions derived from model \eqref{eq:simulation}. 
In Figure~\ref{fig:sim-scatter}, for a given noise level $\omega = 0.1$, we plot the average estimated partial dependence functions of 500 Monte Carlo replications (black solid curves), the 95\% confidence band (i.e. $\pm 1.96\times$ Monte Carlo standard deviation), overlaid with the ground truth functions (gray curves). 

\begin{figure}[t]
\centering
\includegraphics[width=\columnwidth]{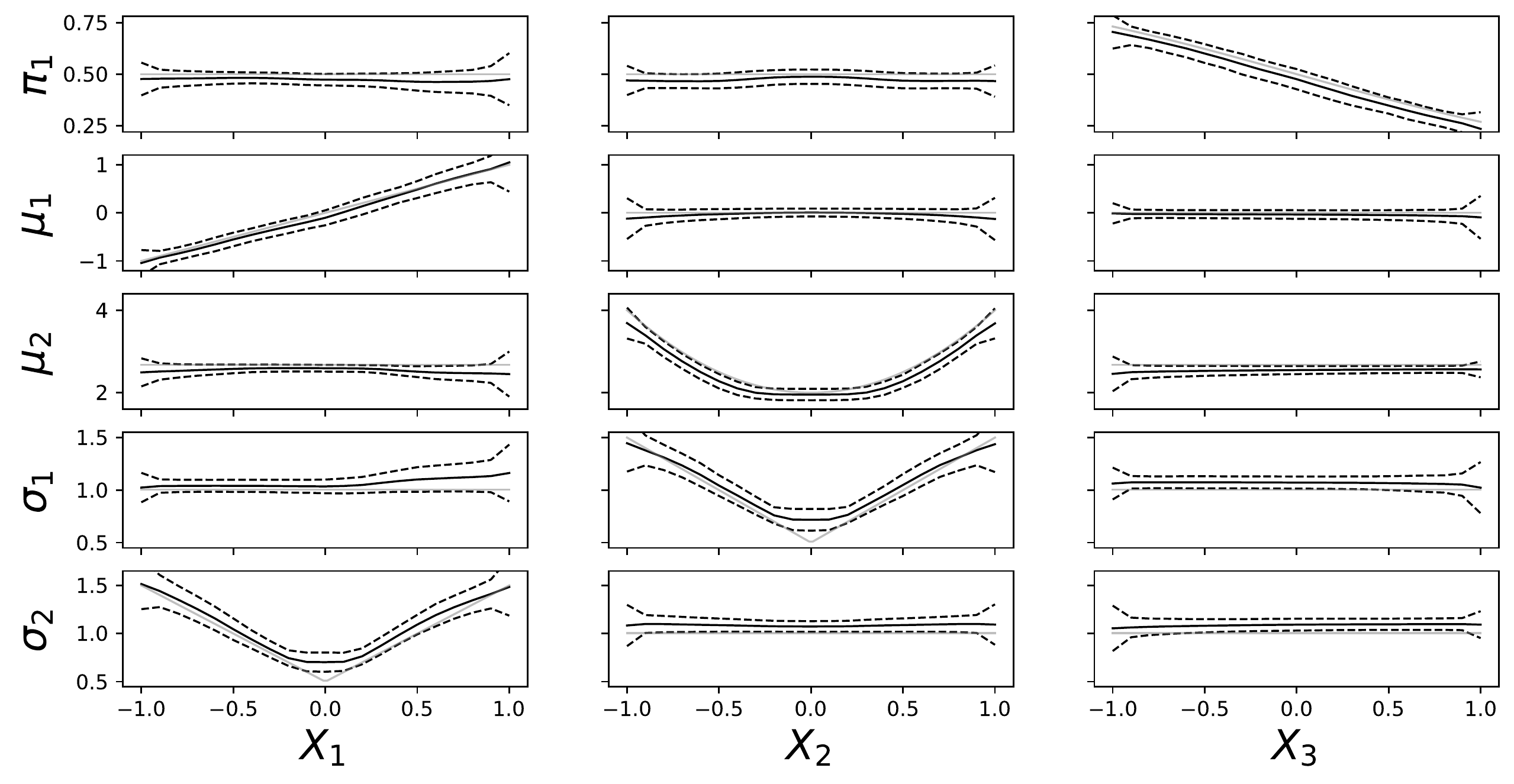}
\caption{Estimated model parameters versus the input scalar variables. The results are based on 500 Monte Carlo replications with noise level $\omega = 0.1$. The black solid curves represent the averaged marginal effect. The dashed curves are the 95\% confidence band. The gray solid curves represent the ground truth.}
\label{fig:sim-scatter}
\end{figure}

\paragraph{Accuracy in estimating conditional distribution outcomes.}
Here we compare how well the four methods estimate the conditional quantile functions of $\mathcal{G}$ given $\mathbf{X}$. We generalize partial dependence plot (PDP) to the quantile functional space to measure the estimation accuracy. At a given quantile level $0 < \rho < 1$, let the target covariate be $\mathbf{X_s}$, and the set of all other covariates be $\mathbf{X_c}$. We define the corresponding functional PDP at point value $\mathbf{X_s} = x_s$ as
$PD_{\mathbf{X_s}}(x_s; \rho) = \int_{x_c \in \mathcal{X}_c} F^{-1}_{\boldsymbol{\tau}(x_s, x_c)}(\rho) f_{\mathbf{X_c}}(x_c) dx_c$, where $f_{\mathbf{X_c}}$ is the marginal density of $\mathbf{X_c}$. Figure \ref{fig:partial-dependence} compares the functional partial dependence plots for the three methods with the ground truth at various quantile levels in $\rho \in \{ 10\%, 30\%, 50\%, 70\%, 90\%\}$. The functional partial dependence plots of the ground truth correspond to the conditional expectation of the functional outputs $\mathbb{E}_{\mathcal{F}}(\mathcal{G} | \mathbf{X})$. The functional partial dependence plots of the three methods were calculated using 500 Monte Carlo replications with noise level $\omega = 0.1$. As shown in Figure \ref{fig:partial-dependence}, our proposed Wasserstein distributional learning (WDL) is capable of capturing the heterogeneity in the partial dependence curves and is closest to the ground truth. The partial dependence curves for Fr{\'e}chet regression are all approximately straight lines of different slopes, due to its linearity assumption. As for the CLR regression, their fitted functional partial dependence curves are also constrained by a linearity assumption after the centered log-ratio transformation. MDN method is not presented here because in the case of GMM with small noise levels it generates similar estimation output with WDL.

\begin{figure}[t]
\centering
\includegraphics[width=\columnwidth]{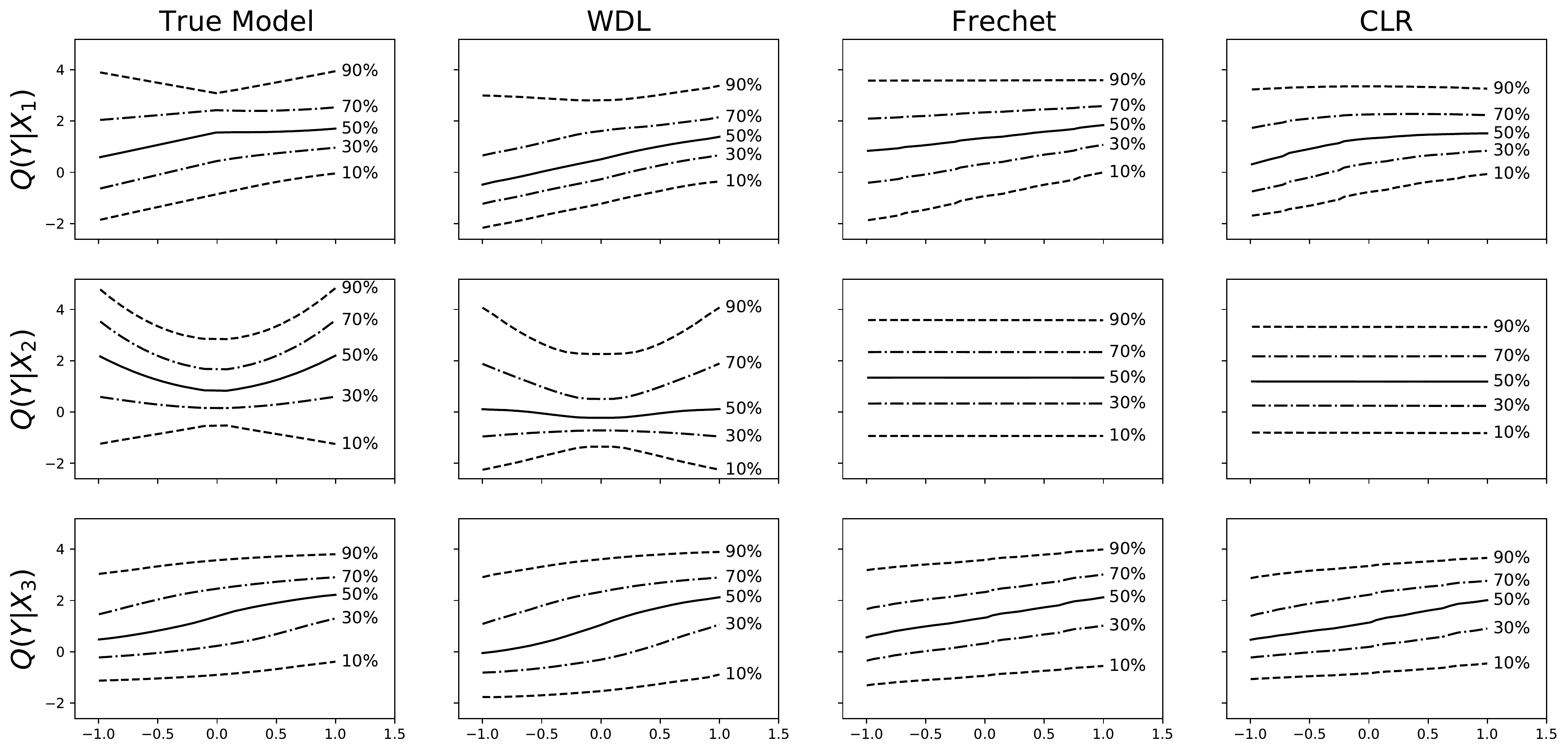}
\caption{Functional partial dependence plot for predicted conditional quantiles versus the input scalar variables. The results are averaged over 500 Monte Carlo replications with noise level $\omega = 0.1$.}
\label{fig:partial-dependence}
\end{figure}

\paragraph{Accuracy in predicting $g_i$'s.} 
We used nested five-fold cross validations to evaluate the prediction performance of WDL in predicting individual density functions $g_i$'s with comparison to the other methods. To numerically measure the discrepancy between the observed quantile functions and their model predictions from the test samples, we defined an approximated Wasserstein distance from a dense array of equally spaced quantile levels $\{0.01, 0.02, ..., 0.99\}$, 
\begin{equation*}
    W_i^2 = \int_{0}^{1} \big(Q_i(s) -  \widehat{Q}_i^{(cv)}(s)\big)^2 ds \approx \frac{1}{100}\sum_{i=1}^{99} \big(Q_i(\frac{i}{100}) -  \widehat{Q}_i^{(cv)}(\frac{i}{100})\big)^2,
\end{equation*}
where $Q_i$ and $\widehat{Q}_i^{(cv)}$ are the observed and predicted quantile functions for a test sample in cross validation, and $W_i^2$ is the Wasserstein prediction loss of the $i$-th observation. We also define an $R^2$-like statistic:
\begin{equation}
    \widehat{R}^2 = 1 - \frac{\sum_{i=1}^N W^2_i \Big/ N}{Var(\mathcal{\mathcal{G}})},
\end{equation}
where the variance of $\mathcal{G}$ is approximated by $Var(\mathcal{\mathcal{G}}) \approx \frac{1}{N} \sum_{i=1}^N \int_{0}^{1} \big(Q_i(s) -  \Bar{Q}(s)\big)^2 ds$ and $\Bar{Q}(s) = \frac{1}{N} \sum_{i=1}^N Q_i(s)$ for any $s \in [0, 1]$. Table 1 summarizes the average Wasserstein loss and R-square at different noise levels ($\omega =0.1, 0.2, 0.5, 1$ and $2$) using a nested five-fold cross validation. The prediction accuracy (measured by average Wasserstein loss) and power (R-square) decline as the noise level increases. At most noise levels ($\omega =0.1, 0.2, 0.5$ and $1$), WDL delivers the best prediction due to its ability to model complicated density output. When the noise level is high ($\omega = 2$), the Fr{\'e}chet regression performs slightly better than the others due to the robustness of its linear model assumption. Another interesting finding is with larger noise level, the performance decay of MDN is more significant than other methods due to overfitting in neural network training. We provide more details about the comparison of WDL and MDN in the Appendix.

\begin{table}
\caption{Predictive performance comparison at different noise levels: Wasserstein loss and $\widehat{R}^2$ (bracket).}
\centering
\begin{tabular}{c c c c c c}
Method & $\omega = 0.1$ & $\omega = 0.2$ & $\omega = 0.5$ &  $\omega=1$ & $\omega=2$ \\ \hline
WDL & \textbf{0.05 (0.92)} & \textbf{0.09 (0.85)} & \textbf{0.30 (0.65)} & \textbf{1.09 (0.34)} & 3.91 (0.02)\\
Fr{\'e}chet & 0.27 (0.54) & 0.30 (0.53) & 0.49 (0.44)& 1.18 (0.28) & \textbf{3.87 (0.03)}\\
CLR   & 0.27 (0.54) & 0.30 (0.53) & 0.50 (0.43) & 1.23 (0.25) & 4.01 (0.00)\\
MDN & 0.13 (0.77) & 0.29 (0.55) & 0.40 (0.54) & 1.21 (0.26) & 5.09 (-0.27)
\end{tabular}
\end{table}


We also compared the three methods using simulations under a linear setting that favors the Fr{\'e}chet regression and CLR approach, where WDL is able to achieve comparable results (see Appendix). This suggests that WDL offers stable performance under numerous settings. We further explored the scenario where the functional outputs are at sparse quantile levels (i.e., $\{0.1, 0.2, ..., 0.9\}$). Under such a setting, CLR and MDN cannot be applied since they require density outcomes, which cannot be easily derived from sparse quantiles. We compared WDL with Fr{\'e}chet regression under the sparse quantile setting and found the results to be comparable to the dense quantile case presented above (see Appendix).

\subsection{Modeling Annual Temperature Distributions}

A fundamental step in climate research is to identify the factors that impact the radiative balance of the planet and are expected to change the temperature distribution. In this section, we apply the proposed Wasserstein distributional learning to understand how the radiative effect of solar irradiance, volcanic eruptions, and CO2, as well as natural climate variability through the El-Ni\~no Southern Oscillation (ENSO) are associated with annual temperature distributions.
These factors have been suggested in climate change literature \cite{fahey2017physical, lewis2017evolution}, and represent natural and human drivers for climate variability and change.  We obtain the daily land-surface average temperatures from Berkeley Earth daily TAVG full dataset \cite{berkeleyearth}, where temperatures are reported as daily anomalies relative to the Jan 1951 $\sim$ Dec 1980 average. We calculate the empirical quantile functions of daily average temperature anomalies for each year between 1880 and 2012 as functional outputs.

The global radiative effects, or ``radiative forcings" used in this example have units of Wm$^{-2}$ and represent the global average energy balance that arises due to changes in atmospheric composition. Here, radiative forcings of solar irradiance, volcanic eruptions, and CO2 as as calculated by the NASA Goddard Institute for Space Studies (GISS) analysis checking the the historical (1850 - 2012) simulation of their dynamical climate model GISS Model E2 \cite{miller2014}. In addition to the three radiative forcing predictors, year-to-year climate variability is summarized through the Ni\~no 3.4 index, a sea surface temperature index that captures the oscillatory of the ENSO system between warm El Ni\~no events and cool La Ni\~na events \cite{mcphaden2020}. Together, these four predictors have been shown to be highly predictive of global annual mean temperature \cite{suckling2017} and are therefore expected to be predictive of the distribution of daily global mean temperature.

\begin{figure}[t]
	\begin{center}
		\includegraphics[width=\textwidth]{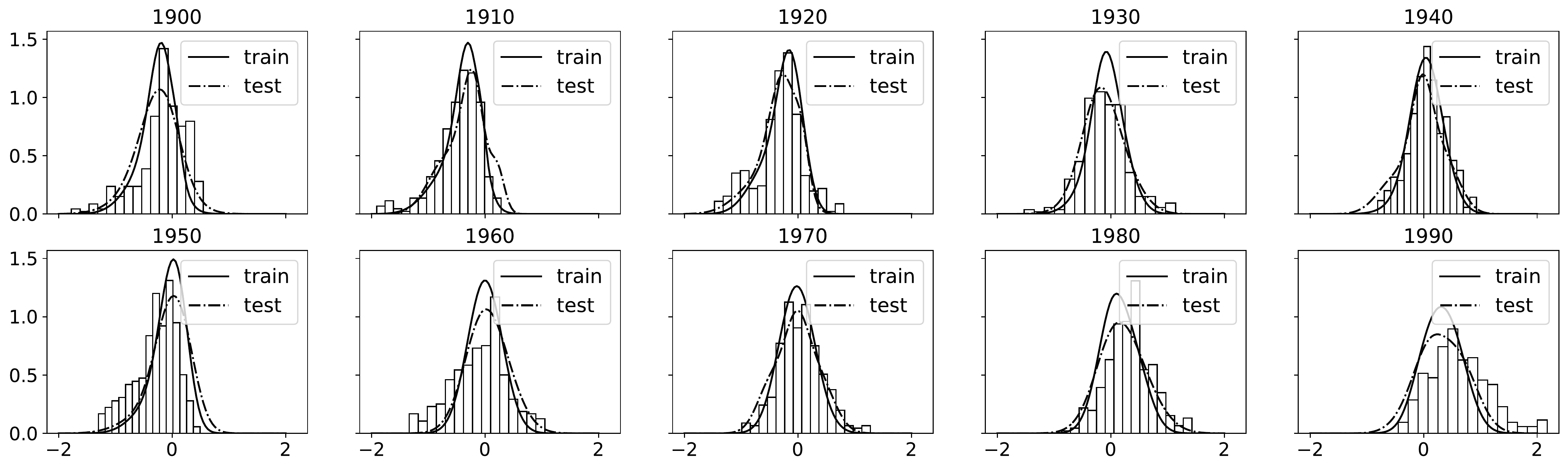}
		\caption{Selected predictions of annual temperature distributions.}
		\label{fig:pdf-pred}
	\end{center}
\end{figure}

\begin{figure}[t]
	\begin{center}
		\includegraphics[width=\textwidth]{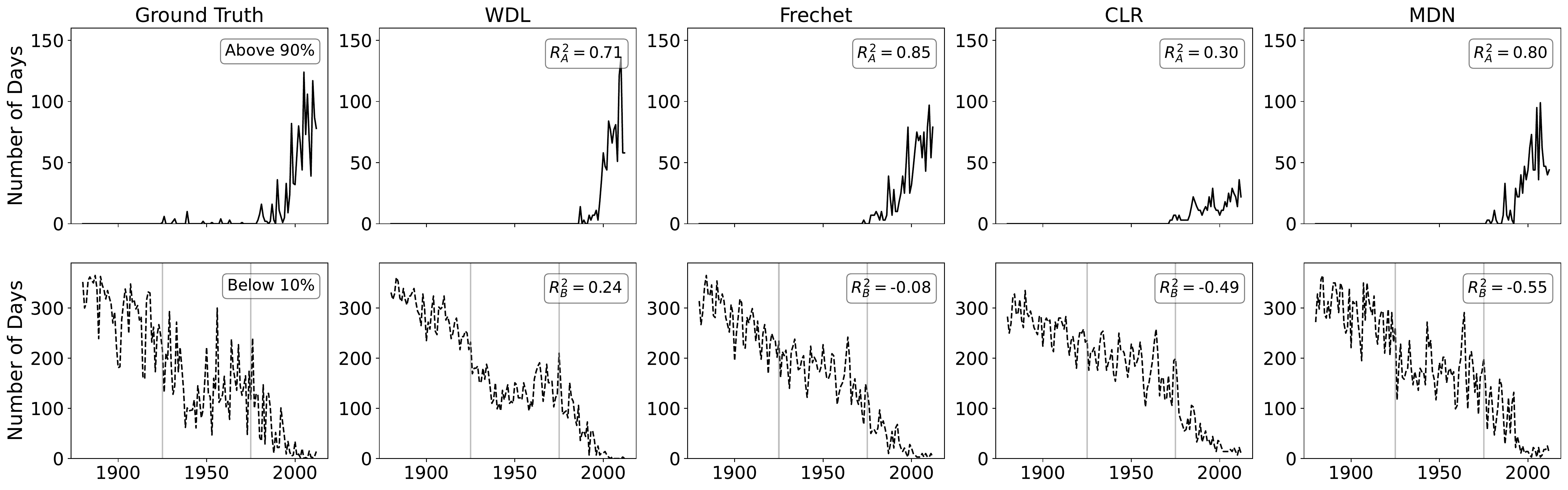}
		\caption{Predicted numbers of days with extreme temperatures. $R^2_A$: Overall test R-squared for predicting the number of days above 90th percentile. $R^2_B$: Local test R-squared for predicting the number of days below 10th percentile in the 1925 - 1975 time window.}
		\label{fig:heatwave-pred}
	\end{center}
\end{figure}
\begin{figure}[t]
	\begin{center}
		\includegraphics[width=\textwidth]{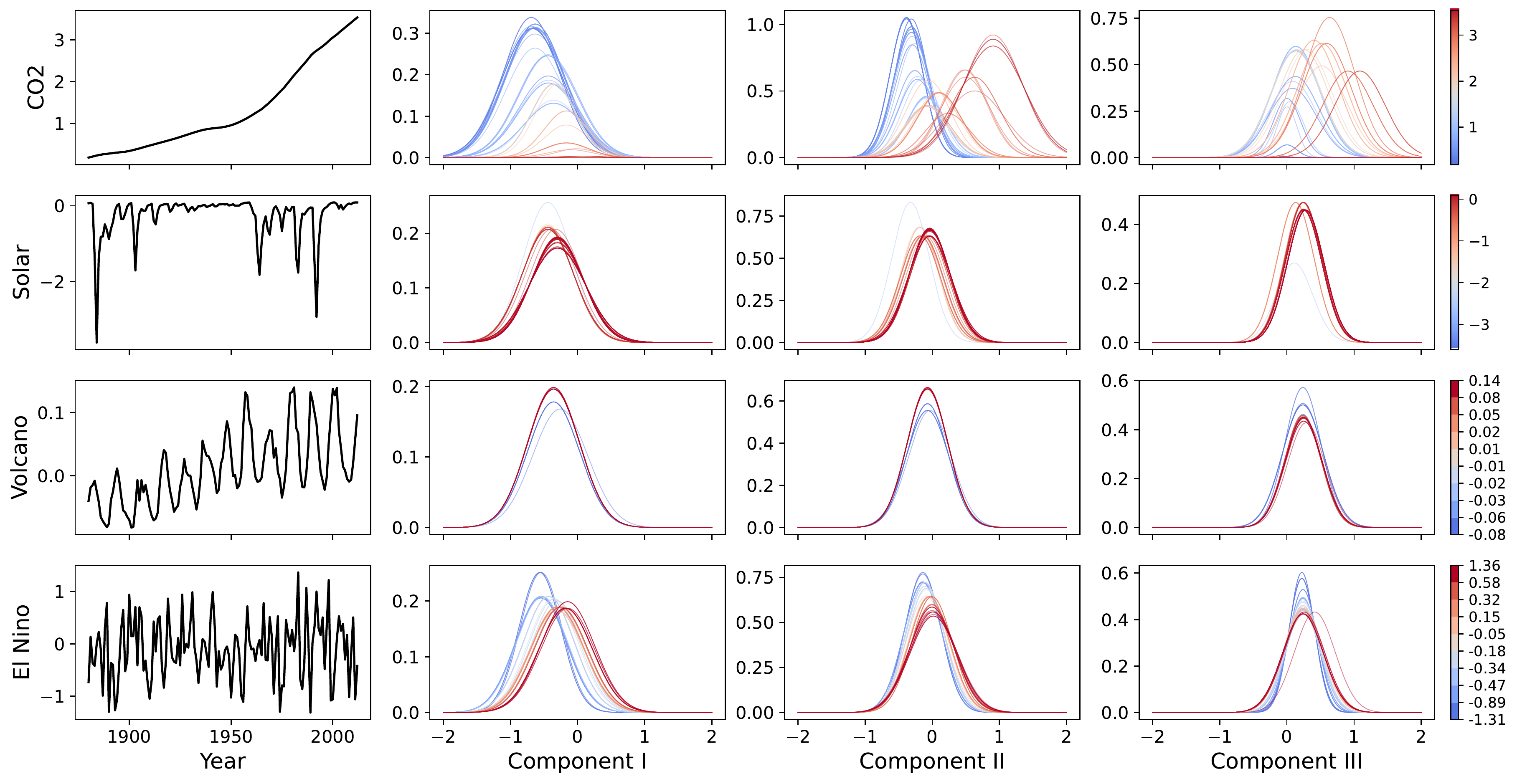}
		\caption{Density curves versus physical drivers. Column 1: temporal trends of physical driver values. Column 2-4: Density curves corresponding to each of three components in the fitted Gaussian mixture, with the color of curves representing the value of the corresponding physical driver.}
		\label{fig:change-pred}
	\end{center}
\end{figure}

Here, we set the number of mixture components as three, which correspond to: low temperatures (Component I), medium temperatures (Component II), and high temperatures (Component III). We fit the proposed Wasserstein distributional learning between the annual temperature quantile functions and the four environmental drivers. To avoid overfitting, we run a nested five-fold cross validation with hyper-parameter selection (learning rate and number of iterations) and calculated the predicted density function for each year when it was in the {\em test} fold. Figure~\ref{fig:pdf-pred} is the histogram of daily temperatures on selected years overlain with the model estimated temperature density curve. The results clearly demonstrate that our method effectively captures the heterogeneity in the functional outputs. Results for each year in the data set can be found in the Appendix. To better evaluate the model performance in predicting extreme temperatures, we calculate the number of days above the 90th percentile daily threshold (high temperatures) or below the 10th percentile daily threshold (low temperatures) for each year. The 10th and 90th percentiles are derived from a 30-year climatological baseline period (1981 - 2010). In Figure~\ref{fig:heatwave-pred}, we visualize the ground truth and the predictions of each method. WDL and Fr{\'e}chet regression achieve the best prediction performance in terms of R-squared. WDL is the only method that is able to capture the cold temperature ``plateau'' and achieves positive R-squared in the 1925 - 1975 time window, which is due to its better characterization for the nonlinear dependence of conditional quantiles (in Appendix).

In Figure~\ref{fig:change-pred}, we visualize the predicted density curves of each component versus the physical drivers. For each given value of a physical driver $\mathbf{X_s} = x_s$, we compute its marginal prediction of the distribution parameters by averaging over the other covariates $\mathbf{X_c}$. Using the component weight $\pi_k$ as an example, the marginal prediction of $\mathbf{X_s} = x_s$ can be formulated as
\begin{equation*}
\Bar{\pi}_k(\mathbf{X_s} = x_s) = \int_{x_c \in \mathcal{X}_c} \hat{\pi}_k(\mathbf{X_s} = x_s, \mathbf{X_c} = x_c) f_{\mathbf{X_c}}(x_c) dx_c.
\end{equation*}
where $\hat{\pi}_k (\cdot)$ is the estimated parameter function, and $f_{\mathbf{X_c}}$ is the marginal density of $\mathbf{X_c}$. Then, based on the average prediction of the parameters $\{\Bar{\pi}_k(\cdot), \Bar{\mu}_k(\cdot), \Bar{\sigma}_k(\cdot) \}_{k = 1, ..., K}$, we calculate the density curve $\{\Bar{\pi}_k(\mathbf{X_s} = x_s)\mbox{N}\{\Bar{\mu}_k(\mathbf{X_s} = x_s), \Bar{\sigma}_k^2(\mathbf{X_s} = x_s)\}$ of each component $k = 1, ..., K$, and depict them using different colors for different values of $\mathbf{X_s} = x_s$. As shown in the figure, CO2 forcing appears to be the most important feature associated with the shifts of the temperature distributions. As the CO2 radiative forcing increases, the mean temperature of all the three components slowly increases as well. Also, ENSO influences the weight and variance of each component, which results in more frequent instances of extreme weather. Radiative forcings due to solar irradiance and volcanic activity did not have a significant impact over the mean temperature. 

\subsection{Modeling Regional Income Distributions}

\begin{figure}[t]
	\begin{center}
		\includegraphics[width=\textwidth]{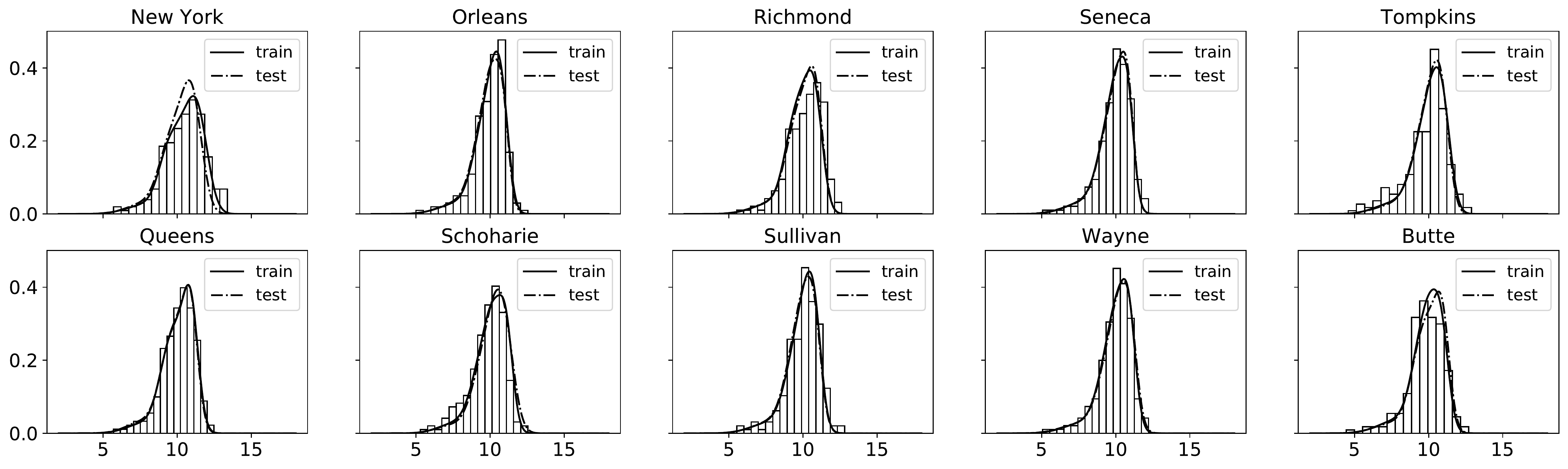}
		\caption{Selected predictions of regional income distributions.}
		\label{fig:income_vis}
	\end{center}
\end{figure}

Modeling income distribution has been a central topic in macroeconomic studies. Several indices, including Gini index, have been used to characterize the income distributions. Traditional analyses focus on one index at a time and models how its mean changes with respect to the covariates. However, the use of a single index only reveals partial distributional associations. On the other hand, simultaneous modeling of multiple indices is usually challenging as their joint distributions might have complex structures due to their unique definitions. In this section, we apply Wasserstein distributional learning to model the regional income distribution of the 167 counties in New York, California and Michigan, from which one could derive multiple indices simultaneously and explicitly study their joint distributions. The income distribution data are from American Community Survey (ACS), which we used with survey weights from 2014 ACS Public Use Microdata Sample (PUMS) \cite{pums} to produce the county-level income distributions. We also collected scalar county-level health indices of the same year (2014) from County Health Rankings \& Roadmaps \cite{health}. Seven important variables were selected for our analysis: Education, Environment, Population, Crime, GDP Per Capita, Diabetes, and Unemployment rate. With all the data in place, the functional regression was conducted at the county level, which means each county served as an independent data point in the regression.

\begin{table}
\caption{Performance comparison in terms of RMSE (and $R^2$) of different indices of income distributions. Results are evaluated on the test folds.}
\centering
\begin{tabular}{c c c c}
Method & Gini Index & Median Income & Poverty Rate\\\hline
WDL & \textbf{0.029 (0.21)} & \textbf{4017.4 (0.37)} & \textbf{0.037 (0.28)}\\
Fr{\'e}chet & 0.052 (-1.56)& 5113.1 (-0.02) & 0.054 (-0.54) \\
CLR   & 0.030 (0.15) & 11065.0 (-3.79) & 0.040 (0.19) \\
MDN & 0.031 (0.08) & 4349.8 (0.26) & 0.040 (0.17)\\
Lasso Regression & 0.030 (0.17) & 4433.3 (0.23) & 0.041 (0.13) \\
Tree Regression & 0.032 (0.06) & 4536.0 (0.20) & 0.039 (0.20)\\
\end{tabular}
\end{table}

In the experiment, we fit the Wasserstein distributional learning to model the association between the regional income distributions and the scalar county health factors. We took the logarithm of income to reduce the skewness. Similar with weather distribution modeling, we illustrated the predicted density for each county when it was in the {\em test} fold. As shown in Figure~\ref{fig:income_vis}, the income distributions vary across counties, and our method is able to capture the distribution accurately. For example, the income distribution in New York has a larger variance, and the income distribution in Orleans has a much higher peak in mode. The WDL framework is able to capture the distinctive features in these two distributions. 

\begin{figure}[t]
	\begin{center}
		\includegraphics[width=\textwidth]{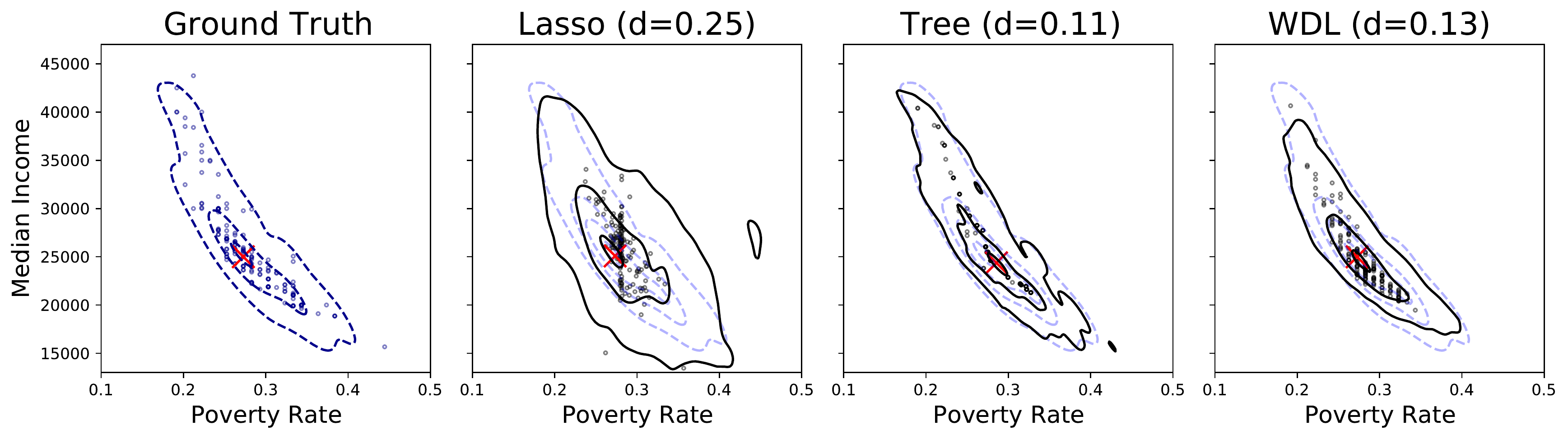}
		\caption{Joint distributions of median income versus poverty rate. Scatterplots overlaid with contour lines. Contour plots of ground truths are visualized as transparent dashed contour curves in each subplot. The modes of the contour plots are marked by red crosses. Wasserstein distances to the ground truth are shown in the sub-figure titles.}
		\label{fig:corr_matrix}
	\end{center}
\end{figure}
From the predicted distributions when the counties were in the test fold, we derived three representative indices commonly used by the economists -- Gini index, median income, and poverty rate, and then compared them with the true values. 
In Table 2, we compared the proposed WDL method with the other methods in terms of estimating individual indices, using RMSE and the conventional $R^2$. In addition to the three comparison methods, we also implemented two methods that directly model the indices: lasso regression and tree regression, which corresponds to the conventional approach of modeling summary statistics in macroeconomics. As shown in Table 2, these index-based methods adequately model the observed indices individually. 
Our WDL algorithm offered the best performance for all the indices, even outperforming the index-based methods (lasso regression and tree regression). In Figure \ref{fig:corr_matrix}, we evaluate the estimated joint distribution of median income versus poverty rate calculated from the predicted income distributions for the counties under study. We compare the estimation using WDL with the ground truth and estimation based on the two index-based methods: lasso regression and tree regression. To get as close to the ground truth as possible, we plotted median income and poverty rate, using the income distributions from individual counties in the data set to calculate. Additionally, we added contour lines that correspond to the joint distribution estimated using kernel density estimation. For each estimation method, we derived the same plot based on model predictions on the test folds with random noises. More specifically, the random noises were sampled from residuals obtained from the training folds. As shown in Figure \ref{fig:corr_matrix}, WDL is able to reconstruct the relationship between indices (summary statistics) without directly modeling them. In particular, predictions based on WDL preserve the true nonlinear association between the two indices without overfitting, offering both stability and flexibility in estimation.

\section{Discussion}
Density-on-scalar functional regression is of interest in many applications. The major obstacle lies in the non-linearity of the density space and the heterogeneity of the density outcomes. Based on the centered log-ratio (CLR) transformation, a number of studies generalized the conventional functional regression algorithms to incorporate the constraints of density functions. However, such generalizations are not applicable to the Wasserstein geometry. As a well-defined metric in the density space, the Wasserstein geometry enjoys multiple nice properties due to its connections to the distribution quantiles. Unfortunately, the estimation under the Wasserstein metric are notoriously challenging, which significantly limits its application in this area.

In this paper, we propose the Wasserstein distributional learning framework, which utilizes the global approximation property of finite mixture models. Compared with prior work, this framework no longer relies on the underlying linear assumption of the model parameters, which greatly increased its flexibility and expressiveness in modeling complicated density outputs. Moreover, efficient optimization algorithms that strongly resemble expectation–maximization (EM) are developed in this work. To our best knowledge, this is the first work that uses majorization to solve the estimation problem associated with the Wasserstein metric and achieves good convergence performance both theoretically and empirically. In the experiments, both simulations and real datasets are used to demonstrate the excellent modeling performance of our method.

The focus of this paper has largely been the parameter estimation of the functional regression model. In the future, however, it would be interesting to develop functional feature importance for the model. One idea is to simply generalize the popular mean decrease impurity (MDI) measure to the case of the Wasserstein loss. The underlying connection between the Wasserstein distance and distribution quantiles is also worth more exploration, and the related functional feature importance at a given quantile level would be helpful to further our understanding of the functional dependence. For practical applications of Wasserstein distributional learning, another interesting future focus is its extension to the more general task of conditional density estimation. Since functional outputs are not primed for analysis, a natural idea is to first apply kernel density estimation to generate distributional outputs, and then feed them to the WDL model. Important topics in this framework include the optimal bandwidth choice for the kernel density estimation.

\newpage
\appendix

\bigskip
\bigskip
\bigskip
\begin{center}
{\LARGE\bf APPENDIX}
\end{center}
\medskip

\renewcommand\thefigure{\thesection.\arabic{figure}}   
\renewcommand\thetable{\thesection.\arabic{table}}
\renewcommand\thealgorithm{\thesection.\arabic{algorithm}}

\section{THEOREM PROOFS}
\subsection{Proof of Lemma 2.1}
\begin{proof}
The proof of this lemma follows the idea of Theorem 6.18. in \cite{villani2003topics}. It suffices to prove that for any probability distribution $g \in \mathcal{P}_2(\mathbb{R})$ and any given constant $\varepsilon > 0$, there exists a finite Gaussian mixture distribution $f = \sum_{k=1}^K \pi_k f_k$, where $f_k = \mbox{N}(\mu_k, \sigma_k^2)$, such that $W_2(f, g) < \varepsilon$.

We prove the above claim in two steps.

First, since $g \in \mathcal{P}_2(\mathbb{R})$, it has a finite second moment, which means $\mathbb{E}_{g}(X^2) < \infty$. Then, there exists a constant $M > 0$ large enough, such that
\begin{equation*}
    \mathbb{E}_{g}[X^2 \mathbf{1}_{\{|X| > M\}}] < \frac{\varepsilon^2}{9}.
\end{equation*}
Cover the compact set $[-M, M]$ by a finite family of balls $\{B(x_k, \varepsilon / 3)\}_{1 \leq k \leq K}$, with centers $x_k \in [-M, M]$, and define
\begin{equation*}
B_k' =
    \begin{cases}
      B(x_1, \varepsilon/3) & \text{if $k=1$}\\
      B(x_k, \varepsilon/3) \backslash \bigcup_{j < k} B(x_j, \varepsilon/3) & \text{if $k=2, ..., K$}
    \end{cases}      
\end{equation*}
Then all $B_k'$ are disjoint and still cover $[-M, M]$.

Define function $J$ on $\mathbb{R}$ by
\begin{equation*}
J(x) =
    \begin{cases}
      x_k & \text{if $x \in B_k' \cap [-M, M]$ for some $k$}\\
      0 & \text{if $x \in \mathbb{R} \backslash [-M, M]$}
    \end{cases}      
\end{equation*}

Then, for any $x \in [-M, M]$, we have $|x - J(x)| < \varepsilon/3$, which leads to the following inequality
\begin{equation}
\begin{split}
    \mathbb{E}_{g}[(X - J(X))^2] 
    &= \mathbb{E}_{g}[(X - J(X))^2\mathbf{1}_{\{|X| > M\}}] + \mathbb{E}_{g}[(X - J(X))^2\mathbf{1}_{\{|X| \leq M\}}] \\
    &\leq \mathbb{E}_{g}[(X-0)^2\mathbf{1}_{\{|X| > M\}}] + \frac{\varepsilon^2}{9}\cdot\mathbb{E}_{g}[\mathbf{1}_{\{|X| \leq M\}}] \\
    &< \frac{\varepsilon^2}{9} + \frac{\varepsilon^2}{9} < \frac{\varepsilon^2}{4}.
\end{split}
\label{ap-eq:coupling}
\end{equation}

Suppose random variable $X \sim g$, we denote $\Tilde{g}$ as the distribution of $J(X)$, saying $J(X) \sim \Tilde{g}$. Then, by the construction of $J$, the distribution $\Tilde{g}$ can be written as $\Tilde{g} = \sum_{k=1}^K \pi_k \delta_{x_k}$, where $\delta_{x_k}$ is the point mass at $x_k$. Moreover, using the definition of the Wasserstein distance, from Equation \eqref{ap-eq:coupling} we have
\begin{equation}
    W_2(g, \Tilde{g}) \leq \sqrt{\mathbb{E}_{g}[(X - J(X))^2] } < \frac{\varepsilon}{2}.
    \label{ap-eq:surrogate}
\end{equation}

Second, we approximate each point mass $\delta_{x_k}$ by a Gaussian distribution $\mbox{N}(x_k, \sigma_k^2)$. Let $f = \sum_{k=1}^K \pi_k \mbox{N}(x_k, \varepsilon^2 / 4)$, then using Theorem 2.4. we can have
\begin{equation*}
\begin{split}
W_2^2(f, \Tilde{g}) 
&= W_2^2\Big(\sum_{k=1}^K \pi_k \delta_{x_k}, \sum_{k=1}^K \pi_k \mbox{N}(x_k, \varepsilon^2 / 4)\Big) \\
&\leq \sum_{k=1}^K \pi_k W_2^2\Big(\delta_{x_k}, \mbox{N}(x_k, \varepsilon^2 / 4)\Big) \\
&= \varepsilon^2 / 4.
\end{split}
\end{equation*}
As a result, we have 
\begin{equation}
W_2(f, \Tilde{g}) \leq \frac{\varepsilon}{2}.    
\label{ap-eq:gmm}
\end{equation}

In conclusion, with $f = \sum_{k=1}^K \pi_k \mbox{N}(x_k, \varepsilon^2 / 4)$ as defined above, combining Equation \eqref{ap-eq:surrogate} and Equation \eqref{ap-eq:gmm} we have
\begin{equation*}
    W_2(f, g) \leq W_2(f, \Tilde{g}) + W_2(\Tilde{g}, g) < \varepsilon,
\end{equation*}
which means the family $\mathfrak{F}_G$ is dense in $(\mathcal{P}_2(\mathbb{R}), W_2)$.
\end{proof}

\subsection{Proof of Theorem 2.2}
\begin{proof}
We prove the theorem in three steps. First, we prove the special case of compact support $\mathcal{X}$. Second, we extend the proof to the case of closed support $\mathcal{X}$. Finally, we prove the theorem for a general $\mathcal{X} \subset \mathbb{R}^p$.

{\bf Step 1.} First, suppose $\mathcal{X}$ is a compact set in $\mathbb{R}^p$, we can prove a stronger version of the theorem, i.e., there exists a step function $\boldsymbol{\tau}(\cdot) \in \mathcal{T}(\mathcal{X}, \boldsymbol{\Theta})$ such that
\begin{equation*}
\begin{split}
W_2\big(H(x), f_{\boldsymbol{\tau}(x)}\big) < \varepsilon, \quad \text{for} \ \forall x\in \mathcal{X}.
\end{split}
\end{equation*}

In fact, by Lipschitz continuity assumption, for $\forall x_1, x_2 \in \mathcal{X}$ and $\|x_1 - x_2 \|_2 < \frac{\varepsilon}{3L}$, we have
\begin{equation*}
    W_2\big(H(x_1), H(x_2)\big) \leq L \cdot \frac{\varepsilon}{3L} = \varepsilon / 3.
    \label{eq: unif}
\end{equation*} 

Let $\delta = \frac{\varepsilon}{6L\sqrt{p}} > 0$, we define $\delta-$box each $x \in \mathcal{X}$ as
\begin{equation*}
B(x, \delta) = \bigotimes_{i = 1}^{p} (x^{(i)} - \delta, x^{(i)} + \delta),
\end{equation*}
which is an open square-shaped neighbourhood covering $x \in \mathcal{X} \subset \mathbb{R}^p$. Also, we have $diam(B(x, \delta)) < \frac{\varepsilon}{3L}$ for each $x$ under the Euclidean distance. Since $\mathcal{X} \subset \bigcup_{x \in \mathcal{X}} B(x, \delta)$, and $\mathcal{X}$ is compact, there exists a finite set $\{x_i\}_{i=1}^N \subset \mathcal{X}$, such that
\begin{equation*}
\mathcal{X} \subset \bigcup_{i=1}^N B(x_{i}, \delta).
\end{equation*}

Define
\begin{equation*}
    \widetilde{B}_i = B(x_i, \delta) \backslash \bigcup_{j < i}B(x_j, \delta),
\end{equation*}
then all $\widetilde{B}_i$ are disjoint and still cover $\mathcal{X}$. With the constructed finite set $\{ \widetilde{B}_i\}_{i = 1}^N$, we define a function $\widetilde{H}(\cdot)$, such that
\begin{equation*}
    \widetilde{H}(x) = H(x_i), \quad \text{if} \ x \in \widetilde{B}_i. 
\end{equation*}
By the definition of $\delta$, we have $W_2\big(H(x), \widetilde{H}(x) \big) < \varepsilon / 3$ for any $x \in \mathcal{X}$. 

By Lemma 2.1, for for each $H(x_i)$, there exists a Gaussian mixture distribution $f_i = \sum_k^{K_i} \pi_{(k; i)} \mbox{N}(\mu_{(k; i)}, \sigma^2_{(k; i)})$ such that $W_2(f_i, H(x_i)) < \varepsilon / 3, i = 1, ..., N$. Let $K = \max_i K_i$, and further decompose each Gaussian mixture distribution into $K$ components. Specifically, for the Gaussian mixture distribution $f_i$ with $K_i < K$ components, we create another Gaussian mixture distribution $\Tilde{f}_i$ by equally dividing weight of the last component into $(K - K_i + 1)$ components. Therefore, without loss of generality, here we simply assume each mixture distribution $f_i$ has the same number of components as $K$, and the components are following the increasing order of their means, saying $\mu_{(1; i)} \leq \mu_{(2; i)} \leq ... \leq \mu_{(K; i)}$. 

With the Gaussian mixture distributions $f_i = \sum_k^{K} \pi_{(k; i)} \mbox{N}(\mu_{(k; i)}, \sigma^2_{(k; i)})$ in place, we construct the following step $\boldsymbol{\tau}(x) = \{ \pi_{k}(x), \mu_{k}(x), \sigma_{k}^2(x)\}_{k = 1}^K$. For each $x \in \mathcal{X}$, it belongs to one and only one $\widetilde{B}_i$. Then, for $k = 1, ..., K$, we let
\begin{equation*}
\pi_{k}(x) = \pi_{(k; i)}, \quad
\mu_{k}(x) = \mu_{(k; i)}, \quad
\sigma_{k}(x) = \sigma_{(k; i)}, \quad \text{if} \ x \in \widetilde{B}_i.
\end{equation*}

By definition, each $\widetilde{B}_i$ is the difference between a series of $\delta-$boxes, which makes their boundaries piecewise axis-parallel. Therefore, the above step function construction $\boldsymbol{\tau}(x)$ is feasible. Let $f_{\boldsymbol{\tau}(x)} = \sum_k^{K} \pi_{k}(x) \cdot \mbox{N}\big(\mu_{k}(x), \sigma_{k}^2(x)\big)$, we have
\begin{equation*}
W_2\big(f_{\boldsymbol{\tau}(x)}, \widetilde{H}(x) \big)  < \varepsilon / 3, \quad \text{for} \ \forall x \in \mathcal{X}.
\end{equation*}

In conclusion, by combining the two parts, we have 
\begin{equation*}
\begin{split}
W_2\big(H(x), f_{\boldsymbol{\tau}(x)}\big) \leq &  W_2\big(H(x), \widetilde{H}(x) \big)  + W_2\big( \widetilde{H}(x), f_{\boldsymbol{\tau}(x)} \big)  \\
< & \varepsilon / 3 + \varepsilon / 3 < \varepsilon, \quad \text{for} \ \forall x \in \mathcal{X}.
\end{split}
\end{equation*}

{\bf Step 2.} Second, we prove the theorem for closed support $\mathcal{X}\subset \mathbb{R}^p$. We choose $M$ large enough, then the integration can be decomposed as
\begin{equation*}
\begin{split}
\int_{x\in \mathcal{X}} W_2(H(x), f_{\boldsymbol{\tau}(x)})d P_\mathbf{X}(x) = &\int_{\{x\in \mathcal{X} | \max_i |x_i| \leq M\}} W_2(H(x), f_{\boldsymbol{\tau}(x)})d P_\mathbf{X}(x) \\
& + \int_{\{x\in \mathcal{X} | \max_i |x_i| > M \}} W_2(H(x), f_{\boldsymbol{\tau}(x)})d P_\mathbf{X}(x).
\end{split}
\end{equation*}

Because $\mathcal{X}\subset \mathbb{R}^p$ is closed, $S = \{x\in \mathcal{X} | \max_i |x_i| \leq M\}$ is compact for any finite $M > 0$. As proved in Step 1, there exist tree models $\boldsymbol{\tau}(x)$ defined over $S$ such that 
\begin{equation*}
    W_2(H(x), f_{\boldsymbol{\tau}(x)}) < \varepsilon / 3, \ \forall x \in S = \{x\in \mathcal{X} | \max_i |x_i| \leq M\}.
\end{equation*}
Thus, we have 
\begin{equation*}
    \int_{\{x\in \mathcal{X} | \max_i |x_i| \leq M\}} W_2(H(x), f_{\boldsymbol{\tau}(x)})d P_\mathbf{X}(x) < \varepsilon / 3.
\end{equation*}

Further, we choose an arbitrary fixed point $x_0 \in S = \{x\in \mathcal{X} | \max_i |x_i| \leq M\}$, and extend $\boldsymbol{\tau}(x)$ to the entire $\mathcal{X}$ by letting $\boldsymbol{\tau}(x) = \boldsymbol{\tau}(x_0)$, for $\forall x \in S = \{x\in \mathcal{X} | \max_i |x_i| > M \}$.

By the Lipschitz continuity assumption, we have
\begin{equation*}
\begin{split}
&\int_{S = \{x\in \mathcal{X} | \max_i |x_i| > M \}} W_2(H(x), f_{\boldsymbol{\tau}(x)})d P_\mathbf{X}(x) \\
\leq &\int_{S} W_2(H(x), H(x_0))d P_\mathbf{X}(x) + \int_{S} W_2(H(x_0), f_{\boldsymbol{\tau}(x)})d P_\mathbf{X}(x) \\
< &\int_{S} L\|x - x_0\|_2d P_\mathbf{X}(x) + \int_{S} W_2(H(x_0), f_{\boldsymbol{\tau}(x_0)})d P_\mathbf{X}(x) \\
< & \ \varepsilon / 3 + \varepsilon / 3, \ \text{as} \ M \longrightarrow \infty.
\end{split}
\end{equation*}

The first term is because of light tail assumption of $P_\mathbf{X}$, and the second term is by the definition of $\boldsymbol{\tau}(x)$.

In conclusion, we have
\begin{equation*}
\begin{split}
\int_{x\in \mathcal{X}} W_2(H(x), f_{\boldsymbol{\tau}(x)})d P_\mathbf{X}(x) = &\int_{\{x\in \mathcal{X} | \max_i |x_i| \leq M\}} W_2(H(x), f_{\boldsymbol{\tau}(x)})d P_\mathbf{X}(x) \\
& + \int_{\{x\in \mathcal{X} | \max_i |x_i| > M \}} W_2(H(x), f_{\boldsymbol{\tau}(x)})d P_\mathbf{X}(x) \\
& < \frac{\varepsilon}{3} + \frac{\varepsilon}{3} + \frac{\varepsilon}{3} = \varepsilon.
\end{split}
\end{equation*}

{\bf Step 3.} Finally, we prove the theorem for general $\mathcal{X} \subset \mathbb{R}^p$. In fact, due to the continuity assumption, we can extend $H(x): \mathcal{X} \longrightarrow \mathcal{P}(\mathbb{R})$ to $\overline{\mathcal{X}}$, the closure of $\mathcal{X}$. We define $\overline{H}(x): \overline{\mathcal{X}} \longrightarrow \mathcal{P}(\mathbb{R})$ as follows
\begin{equation*}
\overline{H}(x) =
    \begin{cases}
      H(x) & \text{if $x \in \mathcal{X}$}\\
      \lim_{y \rightarrow x, \  y \in \mathcal{X}} H(y) & \text{if $x \in \overline{\mathcal{X}} \setminus \mathcal{X}$}
    \end{cases}      
\end{equation*}
Moreover, we can generalize $P_\mathbf{X}$ to $\overline{P}_{\overline{\mathbf{X}}}$ by letting $\overline{P}_{\overline{\mathbf{X}}}(A)= P_\mathbf{X}(A)$ for any $A \subset \mathcal{X}$ and $\overline{P}_{\overline{\mathbf{X}}}(\overline{\mathcal{X}} \setminus \mathcal{X}) = 0$.

By the result of Step 2, we can find $f_{\boldsymbol{\tau}(x)}$ such that the condition is satisfied. Therefore, over $\mathcal{X}$, we have
\begin{equation*}
\begin{split}
&\int_{\mathcal{X}} W_2(H(x), f_{\boldsymbol{\tau}(x)})d P_\mathbf{X}(x) \\
= &\int_{\overline{\mathcal{X}}} W_2(\overline{H}(x), f_{\boldsymbol{\tau}(x)})d \overline{P}_{\overline{\mathbf{X}}}(x) < \varepsilon.
\end{split}
\end{equation*}
\end{proof}

\subsection{Proof of Theorem 2.3}
\begin{proof}
This theorem is an extension of Theorem 5.14 in \cite{van2000asymptotic}. The referenced theorem proves the consistency of M-estimators under regularity assumptions. In our case, the SCGMM estimators are a special case of M-estimators if we generalize the model parameter space $\boldsymbol{\Theta}$ from a Euclidean space into the functional space $\mathcal{T}(\mathcal{X}, \boldsymbol{\Theta})$.

For any $\widehat{\boldsymbol{\tau}} \in \mathcal{T}(\mathcal{X}, \boldsymbol{\Theta})$ such that $\mathcal{M}(\widehat{\boldsymbol{\tau}}) < \infty$, let $U_l \downarrow \widehat{\boldsymbol{\tau}}$ be a decreasing sequence of open balls covering $\widehat{\boldsymbol{\tau}}$ of diameter converging to zero. For any $(x, g) \in \mathcal{X} \times \mathcal{P}_2(\mathbb{R})$ and any open ball $U \subset \mathcal{T}(\mathcal{X}, \boldsymbol{\Theta})$, define $\boldsymbol{m}_U(x, g) = \inf_{\boldsymbol{\tau} \in U} W_2^2(f_{\boldsymbol{\tau}(x)}, g)$ and $\mathcal{M}(U) = \int_{\mathcal{X} \times \mathcal{P}_2(\mathbb{R})}\boldsymbol{m}_U(x, g) d\mathcal{F}(x, g)$. Then, by the construction of $\{U_l\}$ and the continuity of $f_{\boldsymbol{\theta}}$, the sequence $\boldsymbol{m}_{U_l}(x, g) \uparrow W_2^2(f_{\widehat{\boldsymbol{\tau}}(x)}, g)$ for $(x, g) \in \mathcal{X} \times \mathcal{P}_2(\mathbb{R})$ almost surely. Further, by the dominated convergence theorem and the finite assumption of $W_2^2(f_{\widehat{\boldsymbol{\tau}}(x)}, g)$, we have $\int_{\mathcal{X} \times \mathcal{P}_2(\mathbb{R})}\boldsymbol{m}_{U_l}(x, g) d\mathcal{F}(x, g) \uparrow \mathcal{M}(\widehat{\boldsymbol{\tau}}(x)) = \int_{\mathcal{X} \times \mathcal{P}_2(\mathbb{R})}W_2^2(f_{\widehat{\boldsymbol{\tau}}(x)}, g) d\mathcal{F}(x, g) < \infty$.

By definition, $\boldsymbol{\tau}_0 = \argmin_{\boldsymbol{\tau} \in \mathcal{T}(\mathcal{X}, \boldsymbol{\Theta})}\mathcal{M}(\boldsymbol{\tau})$. For any $\boldsymbol{\tau} \neq \boldsymbol{\tau}_0$, due to the uniqueness assumption we have $\mathcal{M}(\boldsymbol{\tau}) > \mathcal{M}(\boldsymbol{\tau}_0)$. Combine this with the preceding paragraph to see that for every $\boldsymbol{\tau} \neq \boldsymbol{\tau}_0$, there exists an open ball $U_{\boldsymbol{\tau}}$ around $\boldsymbol{\tau}$ with $\mathcal{M}(U_{\boldsymbol{\tau}}) > \mathcal{M}(\boldsymbol{\tau}_0)$. The set $B = \{\boldsymbol{\tau} \in S: d(\boldsymbol{\tau}, \boldsymbol{\tau}_0) \geq \varepsilon\}$ is compact and covered by the balls $\{U_{\boldsymbol{\tau}}: \boldsymbol{\tau} \in B \}$. Let $U_{\boldsymbol{\tau}_1}, ..., U_{\boldsymbol{\tau}_p}$ be a finite sub-cover of $B$, then by the law of large numbers,
\begin{equation}
    \min_{\boldsymbol{\tau} \in B} \mathcal{M}_n(\boldsymbol{\tau}) \geq \min_{j=1, ..., p} \mathcal{M}_n(U_{\boldsymbol{\tau}_j}) \xrightarrow{a.s.} \min_{j=1, ..., p} \mathcal{M}(U_{\boldsymbol{\tau}_j}) > \mathcal{M}(\boldsymbol{\tau}_0).
\end{equation}
Therefore, we have
\begin{equation*}
    \liminf_{n} \min_{\boldsymbol{\tau} \in B} \mathcal{M}_n(\boldsymbol{\tau}) > \mathcal{M}(\boldsymbol{\tau}_0) \quad \text{almost surely},
\label{ap-eq:convergence}    
\end{equation*}
which means
\begin{equation}
\mathbb{P}\big(\liminf_{n} \min_{\boldsymbol{\tau} \in B} \mathcal{M}_n(\boldsymbol{\tau}) > \mathcal{M}(\boldsymbol{\tau}_0)\big) = 1.
\end{equation}

If $\widehat{\boldsymbol{\tau}}_n \in B$, then
$\mathcal{M}_n(\widehat{\boldsymbol{\tau}}_n) = \min_{\boldsymbol{\tau} \in B} \mathcal{M}_n(\boldsymbol{\tau})$, which by definition of $\widehat{\boldsymbol{\tau}}_n$ is no larger than $\mathcal{M}_n(\boldsymbol{\tau}_0)$. Thus, for any $n >= 1$, we have

\begin{equation*}
    \{\widehat{\boldsymbol{\tau}}_n \in B \} \subset \{\min_{\boldsymbol{\tau} \in B} \mathcal{M}_n(\boldsymbol{\tau}) \leq  \mathcal{M}_n(\boldsymbol{\tau}_0)\}.
\end{equation*}
On the other hand, we have the following inequality chain for the RHS term,
\begin{equation*}
\begin{split}
     &\limsup_n \mathbb{P}\big\{\min_{\boldsymbol{\tau} \in B} \mathcal{M}_n(\boldsymbol{\tau}) \leq  \mathcal{M}_n(\boldsymbol{\tau}_0)\big\} \\
\leq \  &\mathbb{P} \big(\limsup_n \{\min_{\boldsymbol{\tau} \in B} \mathcal{M}_n(\boldsymbol{\tau}) \leq  \mathcal{M}_n(\boldsymbol{\tau}_0)\}\big) \\
\leq \  &\mathbb{P}\big( \liminf_{n}\min_{\boldsymbol{\tau} \in B} \mathcal{M}_n(\boldsymbol{\tau}) \leq \liminf_{n} \mathcal{M}_n(\boldsymbol{\tau}_0) \big) \\
= \ & \mathbb{P}\big( \liminf_{n}\min_{\boldsymbol{\tau} \in B} \mathcal{M}_n(\boldsymbol{\tau}) \leq  \mathcal{M}(\boldsymbol{\tau}_0) \big) \quad (\text{law of large numbers})\\
= \  & 1 - \mathbb{P}\big(\liminf_{n} \min_{\boldsymbol{\tau} \in B} \mathcal{M}_n(\boldsymbol{\tau}) > \mathcal{M}(\boldsymbol{\tau}_0)\big) \\
= \ & 0.
\end{split}
\end{equation*}
Therefore, the LHS term $\mathbb{P}(d(\widehat{\boldsymbol{\tau}}_n, \boldsymbol{\tau}_0)  \geq \varepsilon) = \mathbb{P}(\widehat{\boldsymbol{\tau}}_n \in B) \rightarrow 0$, which concludes the consistency proof.

\end{proof}

\subsection{Proof of Theorem 2.4}
\begin{proof}
    To prove this theorem, we need to use the alternate definition of the Wasserstein distance. For any two distribution densities $f, g \in \mathcal{P}_2(\mathbb{R})$, the 2-Wasserstein distance $W_2(f, g)$ between them can also be defined
    as 
    \begin{equation*}
        W_2^2(f, g) = \inf_{\gamma \in \Pi(f, g)} \int_{\mathbb{R} \times \mathbb{R}} (x_1 - x_2)^2 \gamma(x_1, x_2) dx_1 dx_2,
    \end{equation*}
    where $\Pi(f, g)$ is the set of joint distributions $\gamma \in \mathcal{P}_2(\mathbb{R} \times \mathbb{R})$ such that for any $(x_1, x_2) \in \mathbb{R} \times \mathbb{R}$ the marginal distributions satisfy $\int_{\mathbb{R}} \gamma(x_1, s) ds = f(x_1)$ and $\int_{\mathbb{R}} \gamma(s, x_2) ds = g(x_2)$. More important, it can be proved the above definition is equivalent with our previous definition of the Wasserstein distance in Section 2, and the $\gamma^*$ achieving the infimum is called the optimal coupling. In our case, since $\mathbb{R}$ is a Polish space, the optimal coupling exists. More details can be found in \cite{villani2008optimal}.
    
	By the existence of optimal coupling, there are $\{ \gamma^*_k \in \Pi(f_k, g_k)\}_{k=1}^K$, such that for $k = 1, ..., K$, the lower bound of the Wasserstein distance is achieved
	\begin{equation*}
	W_2^2(f_k, g_k) = \int_{\mathbb{R} \times \mathbb{R}} (x_1 - x_2)^2 \gamma^*_k (x_1, x_2) dx_1 dx_2.
	\end{equation*}
	
	Define $\gamma^* = \sum_{k = 1}^K \pi_k \gamma^*_k$, then $\gamma^* \in \Pi(f, g)$ since for any $(x_1, x_2) \in \mathbb{R} \times \mathbb{R}$ the marginal distributions satisfy $\int_{\mathbb{R}} \gamma^*(x_1, s) ds = f(x_1)$ and $\int_{\mathbb{R}} \gamma^*(s, x_2) ds = g(x_2)$. By definition,
	\begin{equation*}
	\begin{split}
	W_2^2(f, g) 
	&\leq \int_{\mathbb{R} \times \mathbb{R}} (x_1 - x_2)^2 \gamma^* (x_1, x_2)dx_1 dx_2 \\ 
	&= \sum_{k =1}^K \pi_k \int_{\mathbb{R} \times \mathbb{R}} (x_1 - x_2)^2 \gamma^*_k (x_1, x_2) dx_1 dx_2\\ 
	&= \sum_{k=1}^K \pi_k W_2^2(g_k, f_k).
	\end{split}
	\end{equation*}
	
	Thus, it remains to prove the equality condition. Define function $\mathbf{t}_f^g(x) = G^{-1}\circ F(x)$, where $G^{-1}$ is the QF of $g$, $F$ is the CDF of $f$. The optimal coupling $\gamma^*$ can be expressed as the
	joint distribution of $\big(X, \mathbf{t}_f^g(X)\big)$ with random variable $X \sim f$ \cite{villani2008optimal}. Further, we let $g_k$ be the distribution of $\mathbf{t}_f^g(X_k)$ with random variable $X_k \sim f_k$, then $g_k$ are in the form of Equation (6) in Section 2. Moreover, the joint distribution $\gamma^*_k$ of $(X_k, \mathbf{t}_f^g(X_k))$ is the optimal coupling such that the lower bound of the Wasserstein distance is achieved
	\begin{equation*}
	W_2^2(f_k, g_k) = \int_{\mathbb{R} \times \mathbb{R}} (x_1 - x_2)^2 \gamma^*_k (x_1, x_2) dx_1 dx_2.
	\end{equation*}
	Finally, the following equations hold
	
	\begin{equation*}
	\begin{split}
	W^2_2(f, g) 
	&= \int_{\mathbb{R} \times \mathbb{R}} (x_1 - x_2)^2 \gamma^*(x_1, x_2) dx_1 dx_2 \\ 
	&= \int_{\mathbb{R}} (x - \mathbf{t}_f^g(x))^2 f(x) dx \\ 
	&= \sum_{k=1}^K \pi_k \int_{\mathbb{R}} (x - \mathbf{t}_f^g(x))^2 f_k(x) dx \\
	&= \sum_{k=1}^K \pi_k W_2^2(g_k, f_k),
	\end{split}
	\end{equation*}
	and $g = \sum_{k=1}^K \pi_k \cdot g_k$ is a valid mixture decomposition.
\end{proof}

\section{OPTIMIZATION FRAMEWORK}
We summarize our boosted MM optimization framework in the following Algorithm \ref{WaMiR} and Figure B.1.
\setcounter{figure}{0} 
\setcounter{algorithm}{0} 
\begin{algorithm}[H]
	\caption{Wasserstein Distributional Learning}
	\label{WaMiR}
	\begin{algorithmic}
		\State {\bfseries Input:} data $\mathcal{D} = \{(x_i, g_i)\}_{i=1}^n$, number of components $K$, number of iterations $M$, learning rate $\eta$.
		\State {\bfseries Initialization:} For $k = 1, ..., K$, randomly initialize constant functions
		$\alpha_k^{(0)}(x)$, $\mu_k^{(0)}(x)$, $z_k^{(0)}(x)$ as starting points. 
		\For{$m = 1$ {\bfseries to} $M$}
		\State {\bf 0: [Current values]} For $i = 1, ..., n$ and $K = 1, ..., K$, compute current predictions $\{\alpha_k^{(m-1)}(x_i)\}$, $\{\mu_k^{(m-1)}(x_i)\}$, and $\{z_k^{(m-1)}(x_i))\}$.
		\State {\bf 1: [$f_k^{(m-1)} \mapsto f_k^{(m)}$]} For $i = 1, ..., n$ and $K = 1, ..., K$, find the optimal mixture coefficients $\{\mu_{k, i}^{(m)}, \sigma_{k, i}^{(m)}\}$, and transform them into $\{\mu_{k, i}^{(m)}, z_{k, i}^{(m)}\}$. 
		
		For $k = 1, ..., K$, fit two regression trees $T_{\mu, k}^{(m)}$ and $T_{z, k}^{(m)}$ separately to the differences $\{\mu_{k, i}^{(m)} - \mu_k^{(m-1)}(x_i)\}, \{z_{k, i}^{(m)} - z_k^{(m-1)}(x_i))\}$, and update $\mu_k^{(m)}(x) = \mu_k^{(m-1)}(x) + \eta T_{\mu, k}^{(m)}(x), z_k^{(m)}(x) = z_k^{(m-1)}(x) + \eta T_{z, k}^{(m)}(x)$.
		\State {\bf 2: [$\pi_k^{(m-1)} \mapsto \pi_k^{(m)}$]} For $i = 1, ..., n$ and $K = 1, ..., K$, find the optimal mixture coefficients $\{\pi_{k, i}^{(m)}\}$, and transform them into $\{\alpha_{k, i}^{(m)}\}$.
		
		For $k = 1, ..., K$, fit a regression tree $T_{\alpha, k}^{(m)}$ to the differences $\{\alpha_{k, i}^{(m)} - \alpha_k^{(m-1)}(x_i)\}$, and update $\alpha_k^{(m)}(x) = \alpha_k^{(m-1)}(x) + \eta T_{\alpha, k}^{(m)}(x)$.
		
		\State {\bf 3: [$g_k^{(m-1)} \mapsto g_k^{(m)}$]} For $i = 1, ..., n$, given $\alpha_k^{(m)}(x)$, $\mu_k^{(m)}(x)$, $z_k^{(m)}(x)$, find the optimal mixture decompositions $\{g_{k, i}^{(m)}\}$ of each output distribution $g_i$.
		\EndFor
		
		\State {\bfseries Transformation:} Transform $\{\alpha_k^{(M)}(x), \mu_k^{(M)}(x), z_k^{(M)}(x)\}_{k=1}^K$ back into $\{\pi_k^{(M)}(x)$, $\mu_k^{(M)}(x)$, $\sigma_k^{(M)}(x)\}_{k=1}^K$
		\State {\bfseries Output:} Boosted tree predictions $\widehat{\boldsymbol{\tau}}(x) = \{\pi_k^{(M)}(x), \mu_k^{(M)}(x), \sigma_k^{(M)}(x)\}_{k=1}^K$
	\end{algorithmic}
\end{algorithm}

\begin{figure}[H]
\centering
\begin{tikzpicture}[node distance=1.8cm,
    every node/.style={fill=white, font=\sffamily}, align=center]
  \node (start) [startstop] {Input: $\mathcal{D}$, $K, M, \eta$};
  \node (init) [base, below of=start] {Initialization: $\widehat{\boldsymbol{\tau}}^{(0)}(x)$};
  \node (current) [base, below of=init] {Compute current predictions: $\{\alpha_k^{(m-1)}, \mu_k^{(m-1)}, z_k^{(m-1)} \}_{k=1}^K$};
  \node (param) [base, below of=current] {Find optimal parameters: $\{\alpha_k^{(m)}, \mu_k^{(m)}, z_k^{(m)} \}_{k=1}^K$}; 
  \node (tree) [base, below of=param] {Fit regression trees $T_k$'s to the parameter differences};
  \node (model) [base, below of=tree] {Update model via Boosting: $\widehat{\boldsymbol{\tau}}^{(m)}(x) = \widehat{\boldsymbol{\tau}}^{(m-1)}(x) + \eta \{T_k(x)\}_{k=1}^K$};
  \node (converge) [decision, below of=model, aspect=2.5] {Converge ?};
  \node (end) [startstop, below of=converge, yshift=-1cm] {Output: $\widehat{\boldsymbol{\tau}}^{(M)}(x)$};    

  \draw[->]             (start) -- (init);
  \draw[->]             (init) -- (current);
  \draw[->]             (current) -- (param);
  \draw[->]             (param) -- (tree);
  \draw[->]             (tree) -- (model);
  \draw[->]             (model) -- (converge);
  \draw[->]             (converge) -- node {Yes}
                                   (end);
  \draw[->] (tree.west) -- ++(-1.6,0) -- ++(0, 0.8) -- ++(0, 1) --                
     node[xshift=-1cm, yshift=-0.9cm, text width=2.5cm]
     {Respectively for $\{\alpha, \mu, z\}$}(param.west);
  \draw[->] (converge.east) -- ++(4.8,0) -- ++(0, 3.2) -- ++(0, 4) --                
     node[xshift=1cm, yshift=-3cm, text width=2.5cm]
     {No}(current.east);
\end{tikzpicture}
\caption{Diagram for Fitting SCGMM}
\label{fig:algo-diagram}
\end{figure}
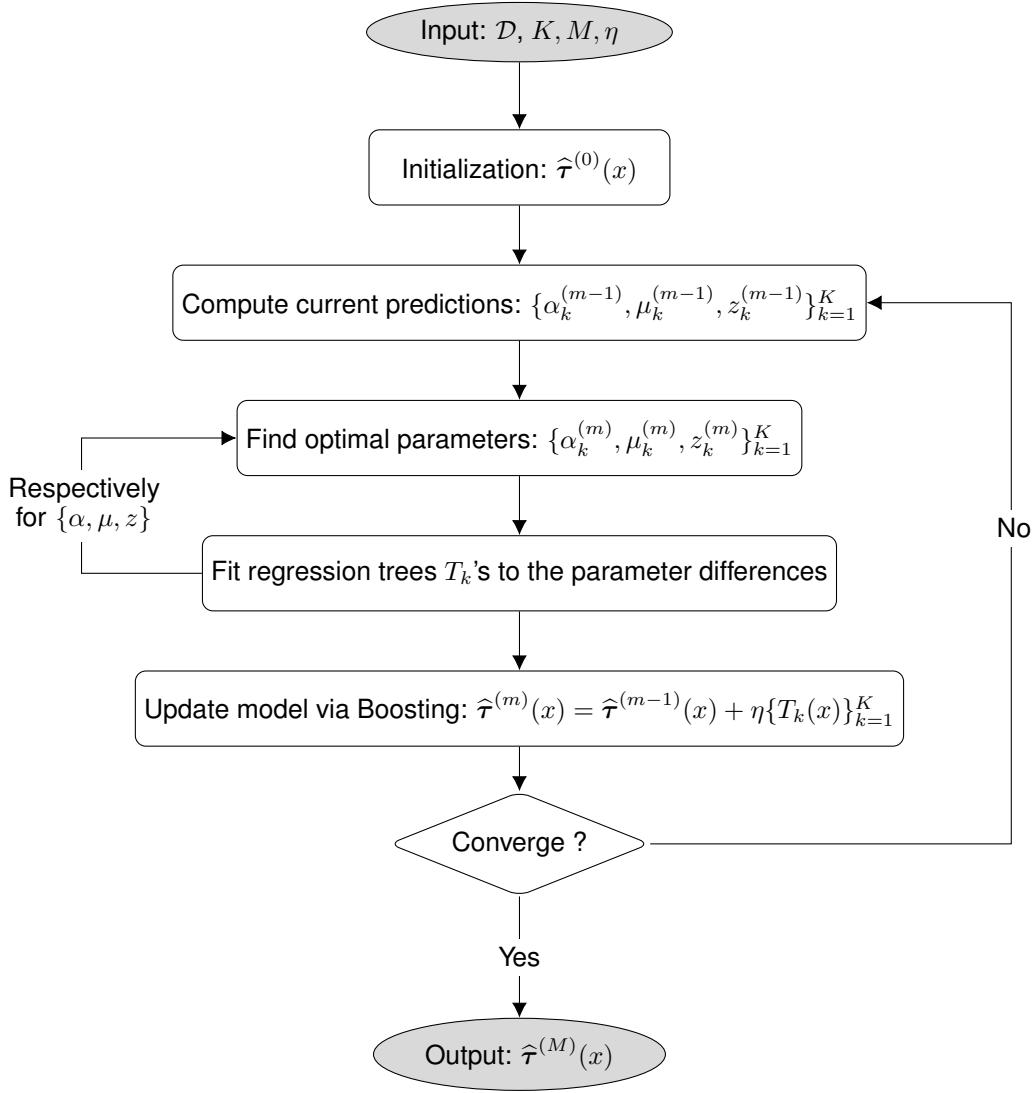

\section{SIMULATION DETAILS}
\setcounter{figure}{0} 
\setcounter{algorithm}{0} 
\setcounter{table}{0} 
In this part, we introduce the simulation details for reproducibility. All the simulations were implemented with \texttt{Python} version 3.6 and \texttt{R} version 4.0.3. 
 
\subsection{Simulation Setup}
The simulation mechanism is illustrated in Equation (7) of Section 3. In the experiment, with predefined parameters $(N, n, \omega)$, density-on-scalar data $\mathcal{D} = \{(x_i, \widehat{g}_i) \}_{i=1}^N$ were simulated as follows

\begin{algorithm}
	\caption{Data Simulation}
\begin{algorithmic}
	\State {\bfseries Input:} $N$-number of samples , $n$-number of data points in each density, $\omega$-noise level.
	\For{$i = 1$ {\bfseries to} $N$}
	\State {\bf [Inputs]} Randomly sample covariate vectors $x_i \sim U[-1, 1]$ and random noises $\varepsilon_i \sim N(0, \omega^2)$.
	\State {\bf [Outputs]} Randomly sample i.i.d. $(y_{i, 1}, ..., y_{i, n})$ from the conditional density $p(Y|X=x_i, \varepsilon_i)$, and construct empirical $\widehat{g}_i$.
	\EndFor
	\State {\bfseries Output:} Density-on-scalar data $\mathcal{D} = \{(x_i, \widehat{g}_i) \}_{i=1}^N$.
\end{algorithmic}
\end{algorithm}

\subsection{Fr{\'e}chet Regression}
The implementation of global Fr{\'e}chet regression was following the algorithm introduced in the reference paper \cite{petersen2019frechet, petersen2021wasserstein}. All the simulations were coded in \texttt{R} using the package \texttt{frechet} \footnote{\url{https://cran.r-project.org/web/packages/frechet/index.html}} developed by the author. In model training, we first calculated the empirical quantile function $\widehat{F}^{-1}_{g_i}$ from the random points $(y_{i, 1}, ..., y_{i, n})$, and then fed them to the function \texttt{GloDenReg}. There is no tuning parameter in global Fr{\'e}chet regression.

\subsection{CLR Regression}
The implementation of the B-spline smoothed density regression with centered log-ratio transformation was based on the sample codes \footnote{\url{https://www.sciencedirect.com/science/article/abs/pii/S0167947318300276}} from the reference paper \cite{talska2018compositional}. Specifically, for a PDF $f \in \mathcal{B}^2(I)$, the centered log-ratio transformation is defined as
\begin{equation*}
    CLR[f](t) = \log f(t) - \frac{1}{\gamma}\int_I \log f(s) ds, \quad \forall t \in I,
\end{equation*}
where $\gamma$ is the normalization constant such that $\int_I CLR[f](t) dt = 0$. Actually, the centered log-ratio transformation defines an isometric isomorphism between the Bayes space $\mathcal{B}^2(I)$ and the Hilbert space $L^2(I)$.  

The hyper-parameters of CLR regression were chosen following the reference paper \cite{talska2018compositional}. With randomly sampled data points $(y_{i, 1}, ..., y_{i, n})$, we first built the histogram of each functional output, of which the optimal number of classes were decided by the Sturges' rule \cite{sturges1926choice}. Because of the output heterogeneity, possible count zeros were replaced by positive posterior expectations with Perks prior \cite{martin2015bayesian}. Afterwards, centered log-ratio transformation was applied to map the density estimations (histograms) into the Bayes space $\mathcal{B}^2(\mathbb{R})$, and B-spline polynomials with equally spaced knots were utilized to smooth the log density curve. To calculate the Wasserstein loss, we transformed the estimated density function $\widehat{f}_i$ back into a quantile function $\widehat{F}^{-1}_i$ at the given 99 equally spaced quantile levels $\{ \tau_l = \frac{l}{100}\}_{l=1}^{99}$ using linear interpolation, and then numerically calculated the Wasserstein loss.

In the sample codes from the reference paper, quadratic B-splines with five equally spaced knots were used. In our implementation, we fine-tuned these hyper-parameters (degree $\in \{2, 3, 4\}$, number of knots $\in \{5, 8, 10\}$) using cross-validation.

\subsection{Mixture Density Network}
\begin{figure}[t]
	\begin{center}
		\includegraphics[width=\columnwidth]{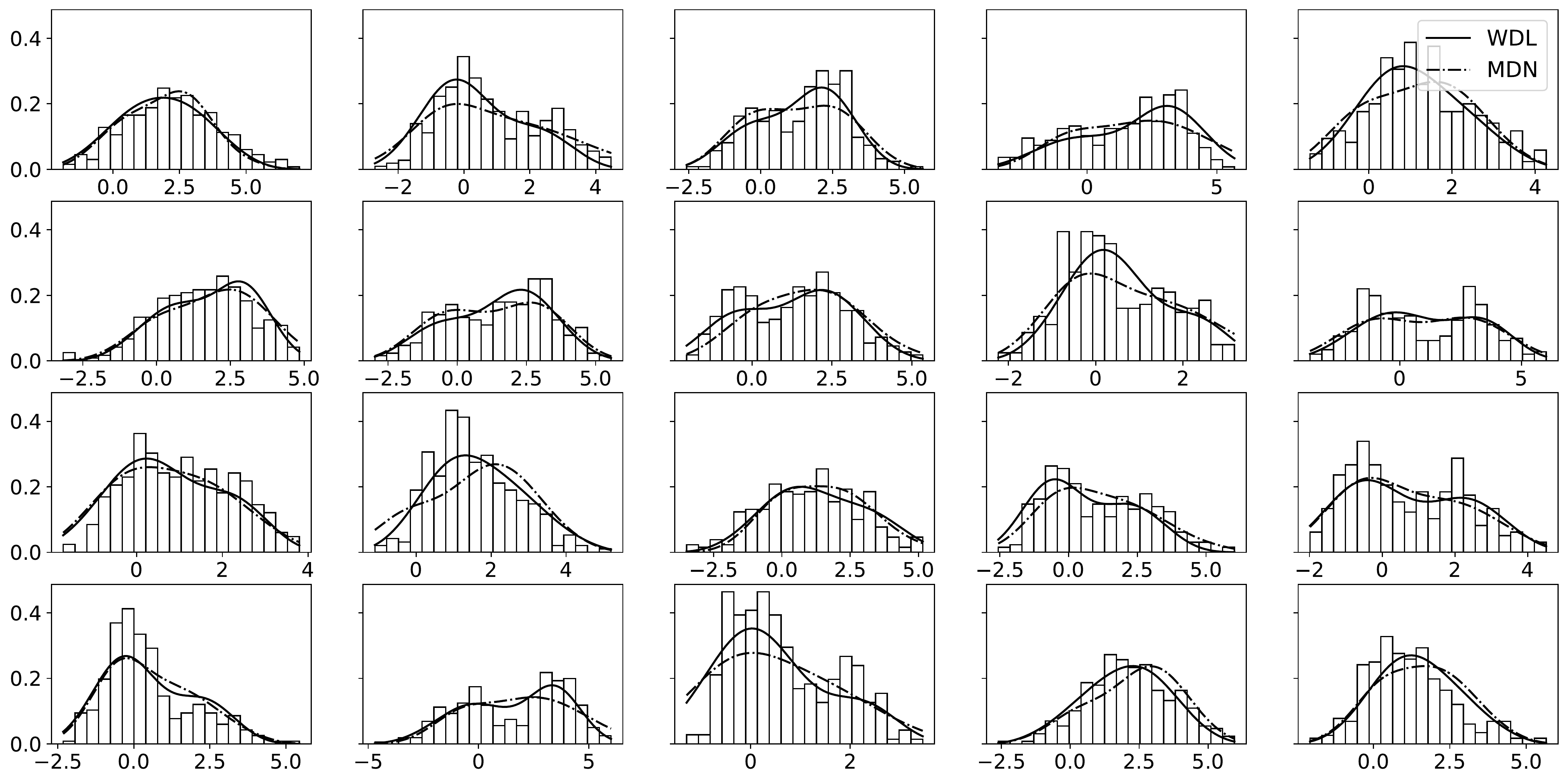}
		\caption{Selected visualizations of density predictions.}
		\label{fig: appendix_density}
	\end{center}
	\vskip -0.2in
\end{figure}

\begin{figure}[ht]
	\begin{center}
		\includegraphics[width=0.7\columnwidth]{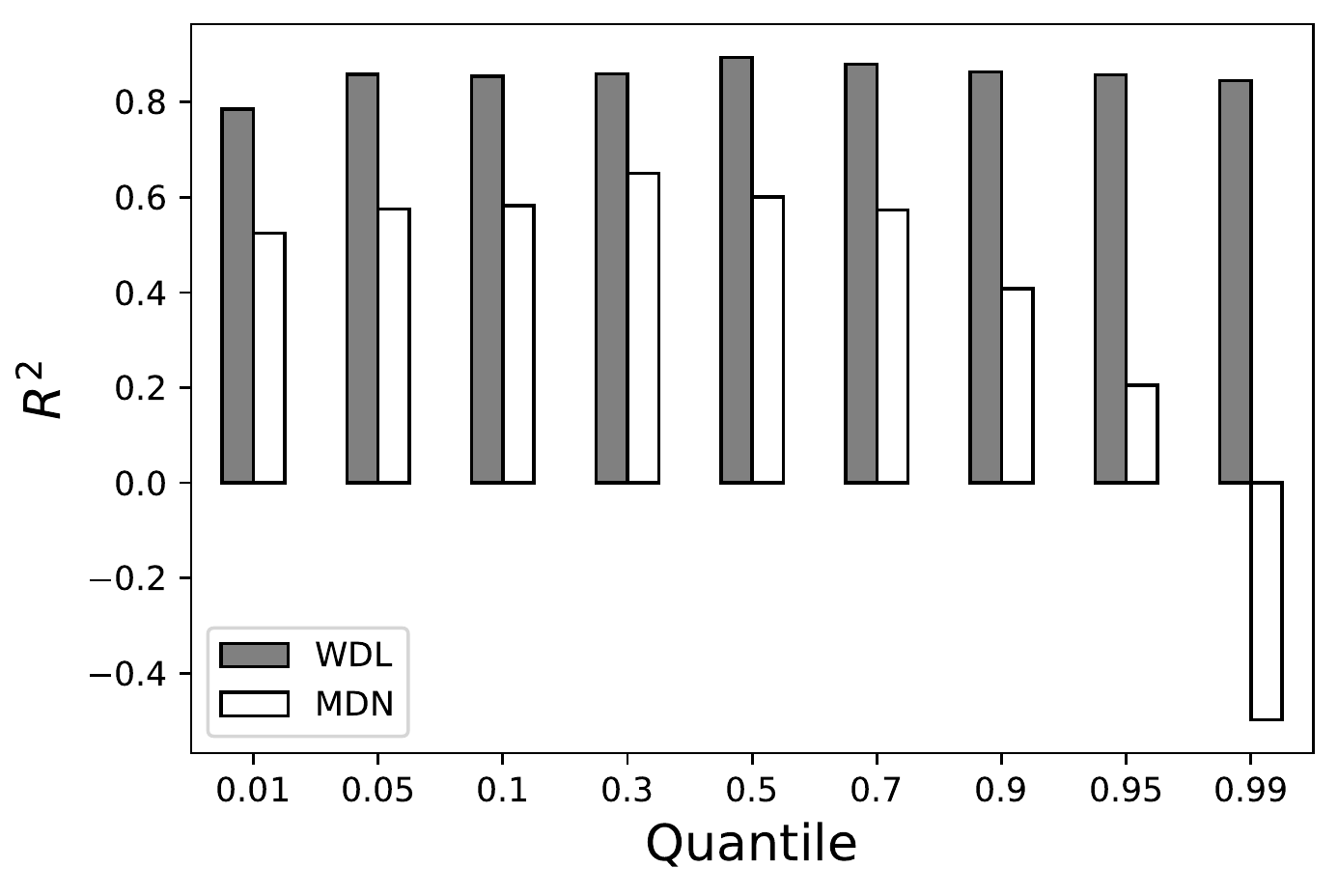}
		\caption{Test R-squared across different quantiles.}
		\label{fig: appendix_RS}
	\end{center}
	\vskip -0.2in
\end{figure}
In this part, we compare the predictions of WDL and the likelihood-based conditional density estimation framework Mixture Density Network (MDN). Instead of modeling the output as density functions, the conditional estimation framework focuses on the estimation of the conditional density $\hat{p}_\theta(y|x)$ from scalar data $\{(x_i, y_i)\}_{i=1}^N$ by maximizing the likelihood function
\begin{equation*}
    \mathcal{L} = \Pi_{i=1}^N p(x_i, y_i) = \Pi_{i=1}^N p(x_i) p_\theta(y_i | x_i).
\end{equation*}

In our simulations, when applying the MDN model, instead of feeding the data $(x_i, g_i)$ with density output $g_i$, we fed the raw data points $(x_i, \{y_{i,j}\}_{j=1}^{n_i})$ into the model, and then estimated the model parameters by maximum likelihood estimation. Quantitative performance comparisons can be found in Table 1. 

Here, we choose the noise level $\omega = 0.2$, under which WDL and MDN have the largest performance difference. We visualize randomly selected predictions of the two methods (curves) as well as the ground truth (histogram) in Figure \ref{fig: appendix_density}. As shown in the figure, the WDL method better captures the general shape of the output density distributions, especially at extreme quantile levels. Also, in Figure \ref{fig: appendix_RS}, we visualize the test R-squared of conditional quantile estimation for these two methods at different quantile levels. Utilizing the Wasserstein distance, the WDL method achieves a more consistent prediction performance across different quantile levels ($R^2 \approx 0.8$). In comparison, the MDN method has a relatively close performance around the median ($\rho = 0.5$), but has poor performance when prediction extreme conditional quantiles ($\rho$ close to 0 and 1).

\subsection{Regression with Sparse Quantiles}
In real-world applications, a common scenario is that the conditional quantiles $F_{g_i}^{-1}(\tau)$ of the functional outputs $\{g_i\}_{i=1}^N$ are only available at a series of sparse quantile values, for instance, $\tau \in  \{0, 0.1, ..., 0.9, 1\}$ in the UK biobank data \footnote{https://biobank.ndph.ox.ac.uk/showcase/field.cgi?id=23000}. To apply the Wasserstein distributional learning framework to these scenarios, additional treatments are essential due to the definition of the Wasserstein loss. To be more specific, as introduced in Section 2, the Wasserstein distance between the two density functions is the integral of their quantile differences. When the dense quantiles are available, e.g. $\tau \in  \{0.01, 0.02, ..., 0.99\}$, the Wasserstein distance can be numerically approximated by the average of all quantile differences, as show below.
\begin{equation*}
    W_2^2(g_1, g_2) = \int_{0}^{1} \big(F_{g_1}^{-1}(s) -  F_{g_2}^{-1}(s)\big)^2 ds \approx \frac{1}{100}\sum_{i=1}^{99} \big(F_{g_1}^{-1}(\frac{i}{100}) -  F_{g_2}^{-1}(\frac{i}{100})\big)^2.
\end{equation*}
While, such approximation is far from accurate when the quantiles are sparse, e.g. $\tau \in  \{0.1, 0.2, ..., 0.9\}$. Similarly, to measure the discrepancy between functional outputs, we can still define the quasi-Wasserstein loss as 
\begin{equation*}
    \widetilde{W}_2^2(g_1, g_2) = \frac{1}{10}\sum_{i=1}^{9} \big(F_{g_1}^{-1}(\frac{i}{10}) -  F_{g_2}^{-1}(\frac{i}{10})\big)^2.
\end{equation*}
However, to the best of our knowledge, there is no efficient algorithm for solving this optimization problem when the model family is Semi-parametric Conditional Gaussian Mixture Models (SCGMM) and the quantiles levels are sparse.

\begin{table}[t]
\caption{Performance comparison in terms of Wasserstein loss and R-squared (bracket) when the quantiles are sparse.}
\begin{center}
\begin{tabular}{rrrrr}
Method & $\varepsilon = 0.1$ & $\varepsilon = 0.3$ & $\varepsilon=0.5$ \\\hline
WDL & \textbf{0.072 (0.84)} & \textbf{0.132 (0.64)} & \textbf{0.216 (0.36)}\\
Fr{\'e}chet & 0.226 (0.51)& 0.247 (0.35)& 0.292 (0.16)
\end{tabular}
\end{center}
\end{table}
Practically, a solution to address this issue is to apply a linear interpolation. Over the training set, we can augment the sparse quantiles into a dense array using linear interpolation. It should be noted that the augmented quantiles naturally satisfy the non-crossing constraints since linear interpolation keeps the monotonicity of the quantile function. Then, the WDL framework can be fitted over the training set using the augmented dense quantiles. Optimization convergence is theoretically guaranteed in Section 2. Last, but not least, we can take the predicted sparse quantiles as the functional output, and the test error can be evaluated over the test set at the given quantile levels.

Using the same simulation setup as Section 3, we rerun the experiments with sparse quantiles $\tau \in  \{0, 0.1, ..., 0.9, 1\}$, and compare the performance of WDL with Fr{\'e}chet regression. CLR regression and MDN are ignored here because under the sparse quantiles, the estimated density would be highly unstable. As shown in Table C.1, the WDL performance is similar with that in the case of dense quantiles (Table 1), and is significantly better than Fr{\'e}chet regression.

\subsection{Simulations in The Linear Case}
In this part, we present simulation results in a simpler setup, and prove our proposed method is flexible under different settings. The simulation setup of this section follows the experiment in \cite{petersen2019frechet}. In this case, the quantile values of the functional output are linear functions of the input variables, which is the underlying assumption of global Fr{\'e}chet regression. Results show our proposed Wasserstein distributional learning framework is able to achieve comparable performance even when the data are simulated in a different way from the model assumption.

To simulate the functional regression data, the responses $Y$ are distributions represented by quantile functions $Q(Y)$ and the predictors are random vectors $X \in \mathbb{R}^3$. For any given quantile level $0 < \tau < 1$, the regression function is 
\begin{equation*}
    Q^{-1}_Y(\tau) = (\mu_0 + \beta^\top x) + (\sigma_0 + \gamma^\top x) \cdot \Phi^{-1}(\tau),
\end{equation*}
where $\Phi$ is the standard normal distribution function. $\mu_0, \sigma_0 \in \mathbb{R}$, $\beta, \gamma \in \mathbb{R}^3$, satisfy $\sigma_0 + \gamma^\top x > 0$ for all $x$. In fact, this simulation scenario corresponds to cases in which the response functions are normal distributions with linear parameters on average.

In the simulation, the functional response $Y$ is generated conditional on $X$ by adding noise to the quantile functions. For each input $X = x$, the distribution parameters $(\mu, \sigma)$ are independently sampled from $p(\mu | X = x) = \mbox{N}(\mu_0 + \beta^\top x, v_1)$ and $p(\sigma | X = x) = Gam((\sigma_0 + \gamma^\top x)^2 / v_2, v_2 / (\sigma_0 + \gamma^\top x))$, and the corresponding functional output is $Y = \mu + \sigma \Phi^{-1}$. Specifically, we set the parameters as $\mu_0 = 0, \sigma_0 = 3, v_1 = 0.25, v_2 = 1, \beta = (1, -1, 3)^\top, \gamma = (0.1, 0.2, 0.3)$, and simulated $N = 200$ random distributions, which are each represented by $n = 300$ random data points.

\begin{figure}[ht]
	\begin{center}
		\includegraphics[width=\columnwidth]{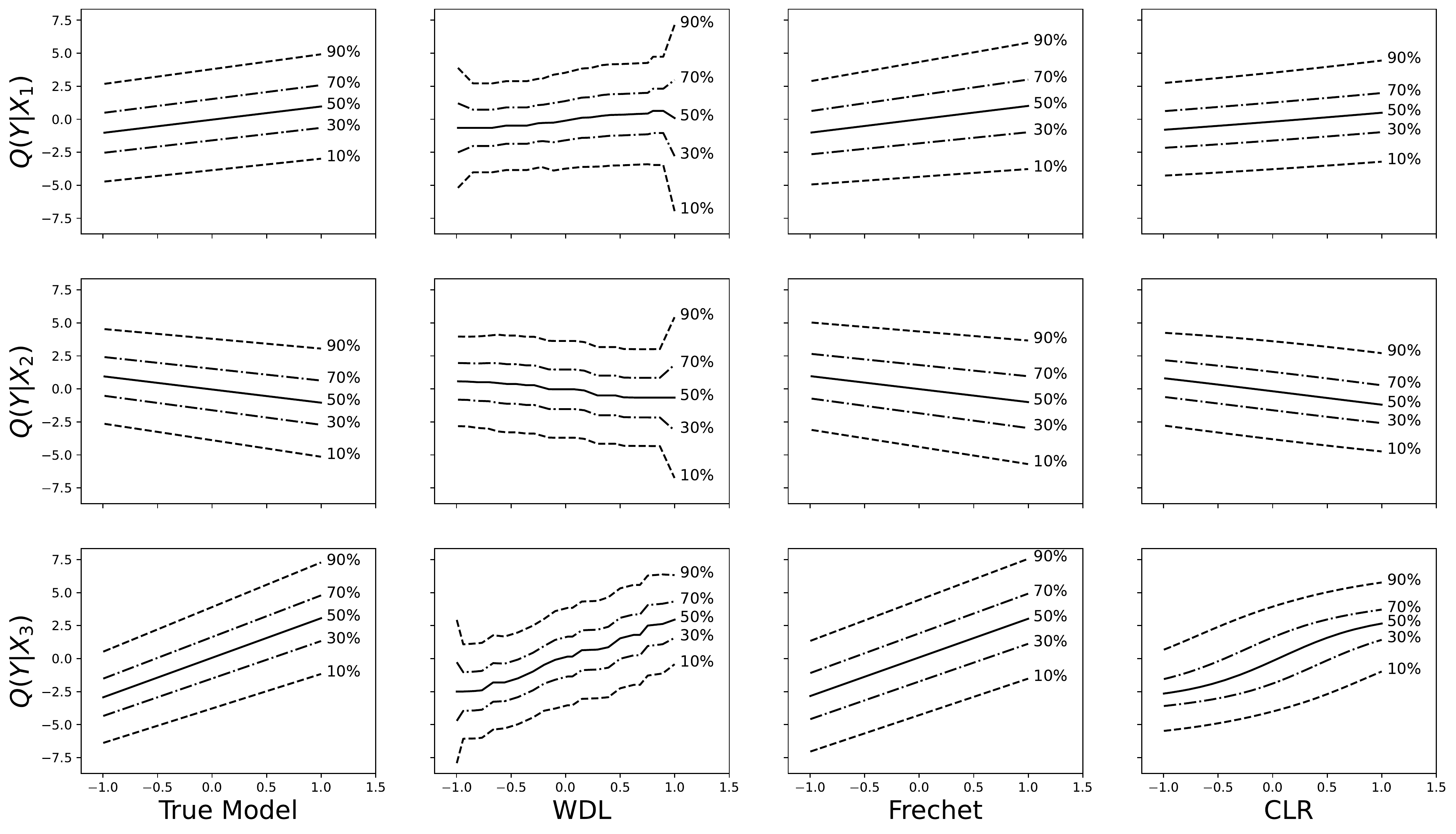}
		\caption{Functional partial dependence plot for the three methods.}
		\label{fig:partial-dependence-appendix}
	\end{center}
	\vskip -0.2in
\end{figure}

The functional partial dependence plots on test sets are observed in Figure \ref{fig:partial-dependence-appendix}. Results show that the proposed Wasserstein distributional learning has stable performance under different simulation settings.

\section{REAL-WORLD APPLICATIONS}
In this appendix section, we provide more details for the real-world applications.

\subsection{Climate Modeling}
\setcounter{figure}{0}   

\begin{figure*}[ht]
	\begin{center}
		\includegraphics[width=0.7\textwidth]{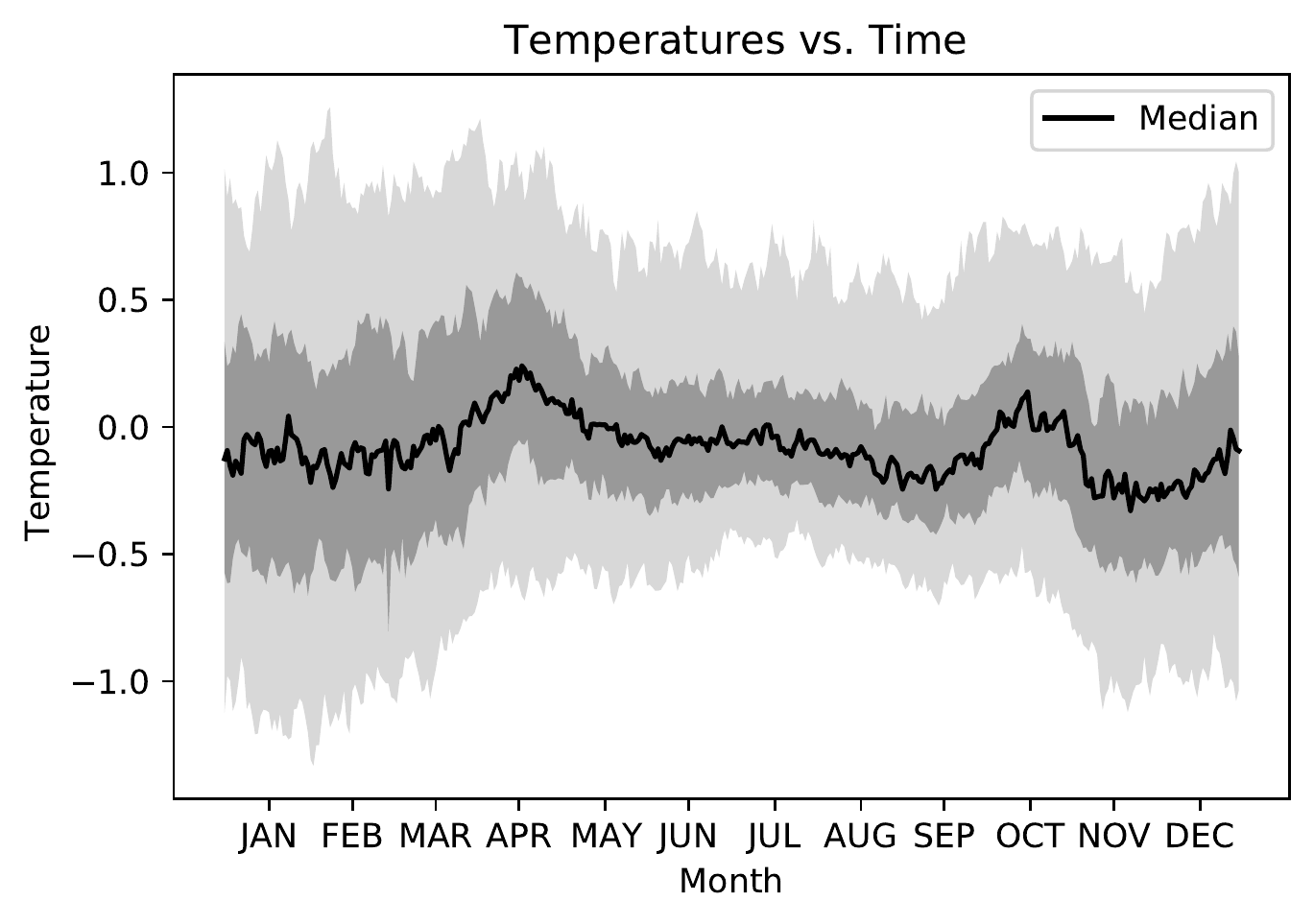}
		\caption{Daily trend of global average temperatures. The solid curve represents the median temperature of each day. The dark gray band represents the 30\% and 70\% quantiles. The light gray band represents the 10\% and 90\% quantiles.}
		\label{fig:appendix_daily_plot}
	\end{center}
\end{figure*}

In Figure \ref{fig:appendix_daily_plot}, we visualize the daily trend of global average temperatures. Using data from 1880 to 2012, we calculate the temperature quantiles (10\%, 30\%, 50\%, 70\%, 90\%) for each day. An interesting finding is that spring and autumn have a higher temperature in general than the other two seasons. Also, the temperature variability in summer is smaller than the other seasons. A potential explanation would be that the temperatures were calculated by averaging the records from multiple weather stations both from the north and south hemisphere. As a result, the averaged temperatures would display a more complicated trend since the north and south hemispheres always have different seasons. 

\begin{figure*}[t]
	\begin{center}
		\includegraphics[width=\textwidth]{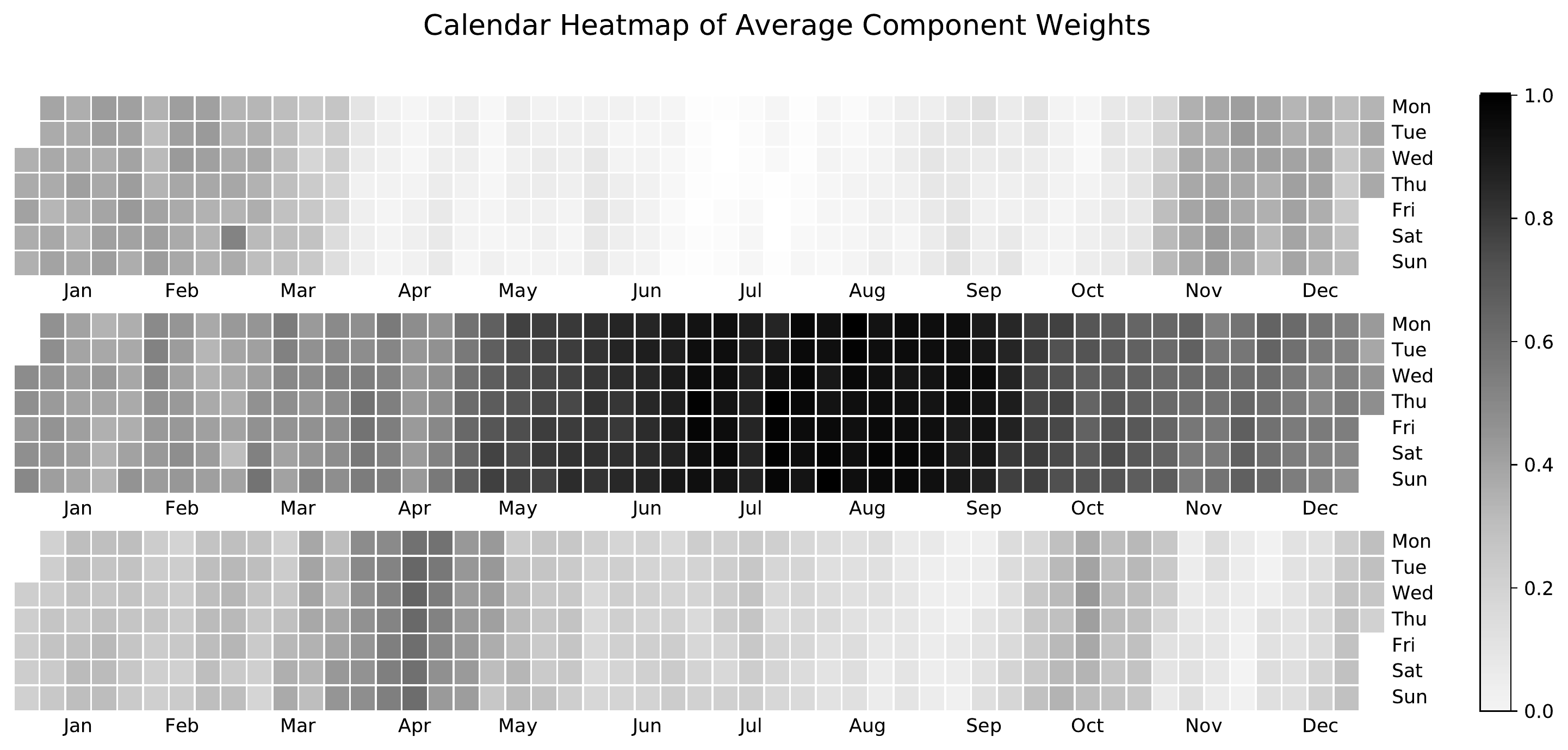}
		\caption{Average Component Weights of Each Day. First row: Component I; Second row: Component II; Third row: Component III. Component weights are represented using different colors.}
		\label{fig:appendix_component_weight}
	\end{center}
\end{figure*}

We also visualize the average component weights of each day in Figure \ref{fig:appendix_component_weight}. Specifically, for each day from 1880 and 2012, we calculate the weight of each component using the predicted WDL model, and then average them across years. Finally, we plot them using a calendar heatmap with each grid representing a day in the year 2020 (we chose 2020 because it is a leap year with 366 days).

\begin{figure}[ht]
	\begin{center}
		\includegraphics[width=\textwidth]{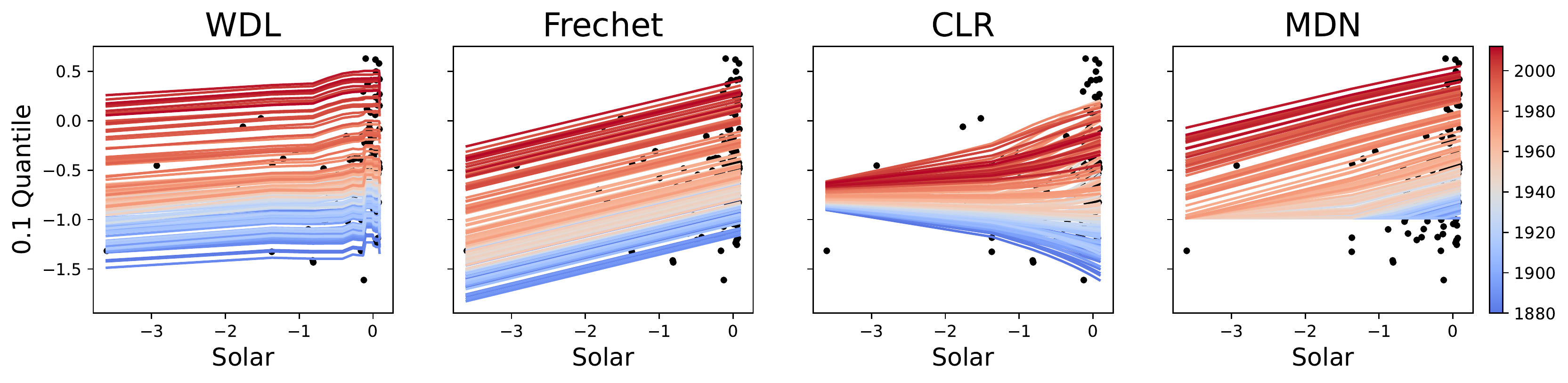}
		\caption{ICE plots of conditional temperature quantiles (10\%) by solar irradiance. True conditional quantiles vs. solar irradiance from raw data are represented as black dots.}
		\label{fig: appendix_ICE}
	\end{center}
\end{figure}

In Figure \ref{fig: appendix_ICE}, we make the Individual Conditional Expectation (ICE) plots for each method. From the figure, WDL and Fr{\'e}chet regression are the only two methods that can give unbiased estimations of conditional quantiles due to their choices of Wasserstein loss. Compared with Fr{\'e}chet regression, WDL performs better when there exists nonlinearity in the conditional dependence. These findings explain the phenomenon that WDL is the only method that can predict the ``cold temperature plateau'' between 1940 and 1960 in Figure 5. In Figure \ref{fig: appendix_temp_1} to Figure \ref{fig: appendix_temp_3}, we visualize the predicted annual temperature distributions for each method from 1880 to 2012. Also in those figures, WDL captures the tail behavior more accurately than the others.

\begin{figure}[h]
	\begin{center}
		\includegraphics[width=\textwidth]{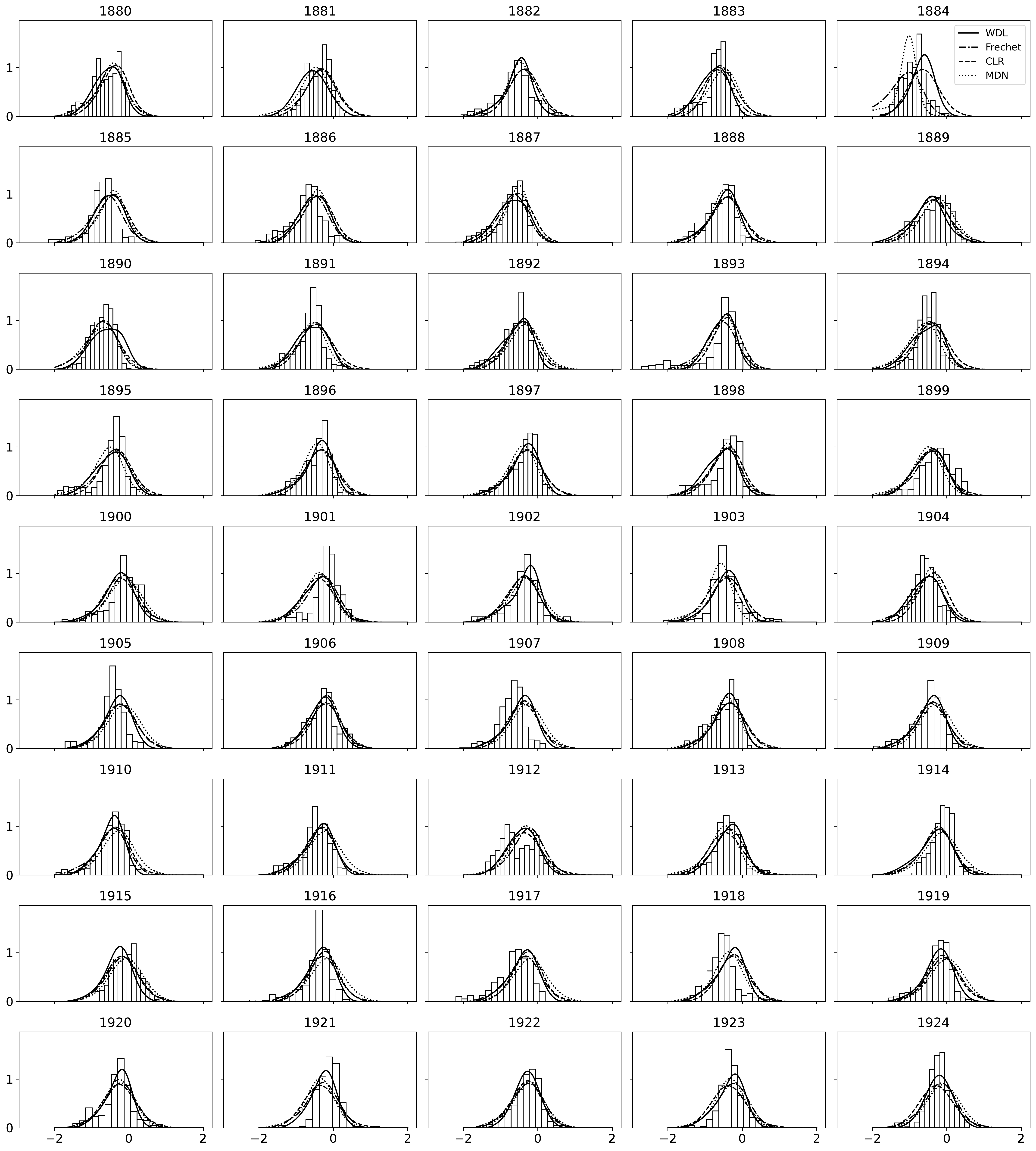}
		\caption{Predictions of annual temperature distributions (Part I).}
		\label{fig: appendix_temp_1}
	\end{center}
\end{figure}

\begin{figure}[h]
	\begin{center}
		\includegraphics[width=\textwidth]{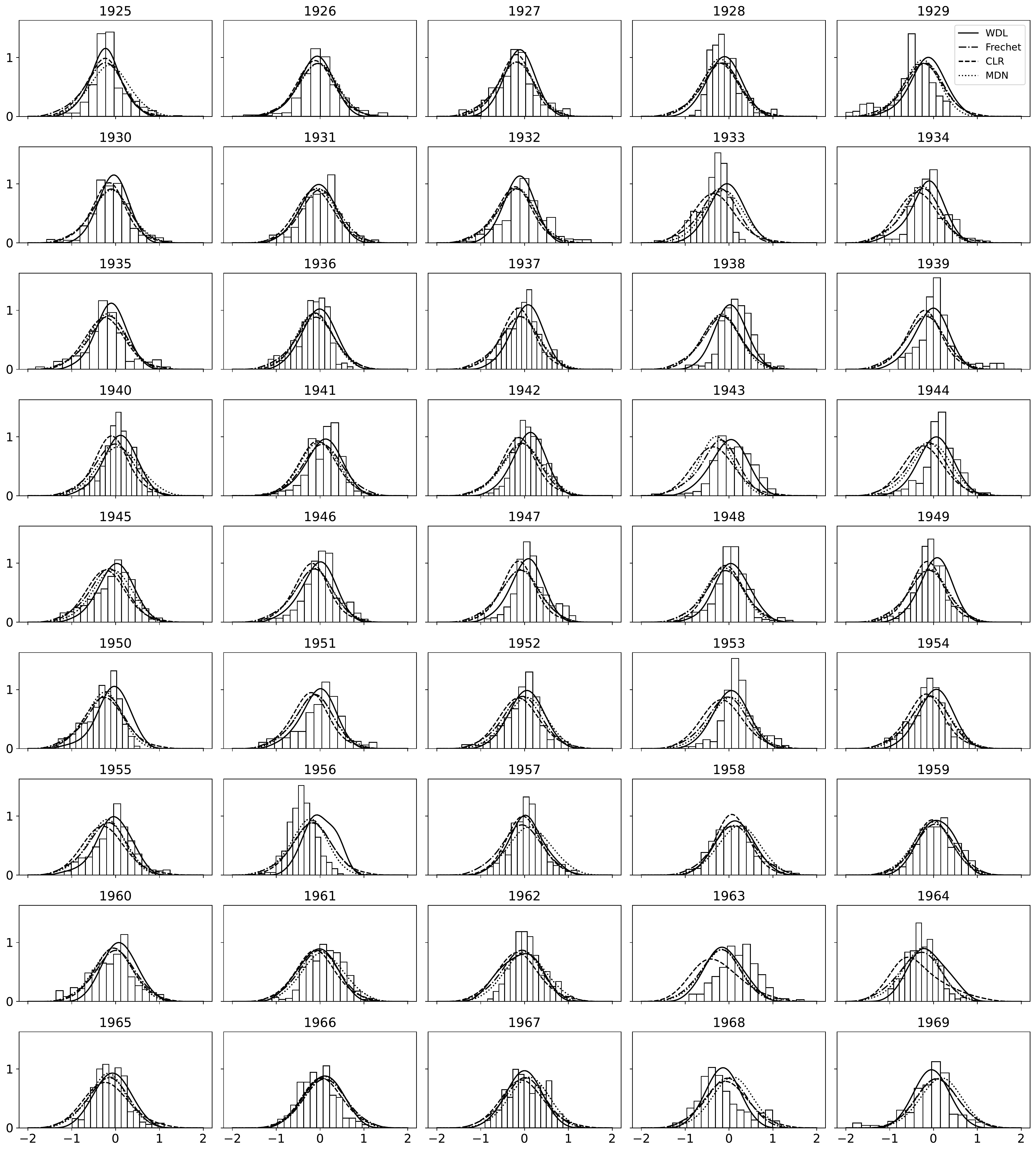}
		\caption{Predictions of annual temperature distributions (Part II).}
		\label{fig: appendix_temp_2}
	\end{center}
\end{figure}

\begin{figure}[h]
	\begin{center}
		\includegraphics[width=\textwidth]{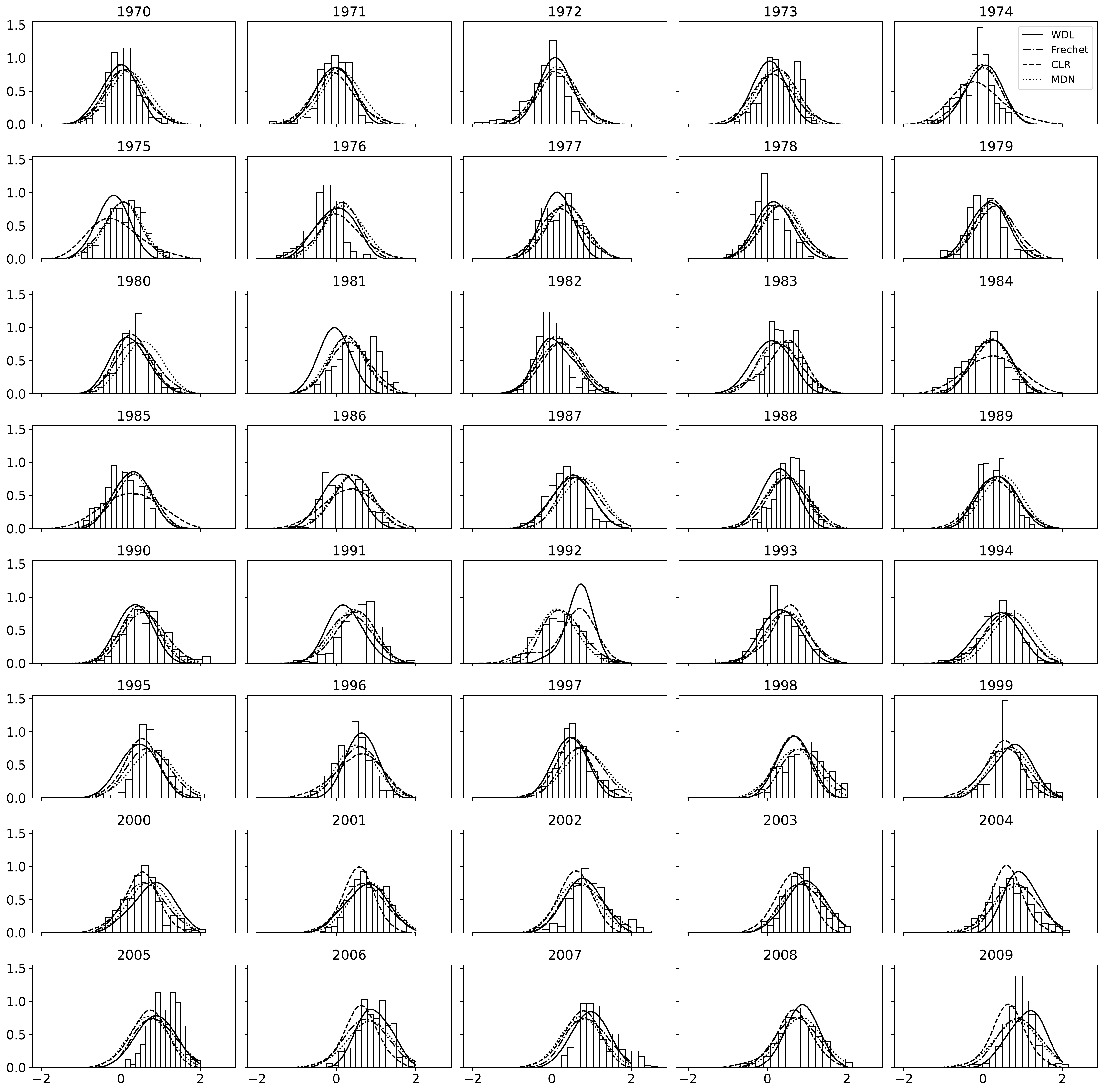}
		\caption{Predictions of annual temperature distributions (Part III).}
		\label{fig: appendix_temp_3}
	\end{center}
\end{figure}

\subsection{Income Modeling}
In this part, we make the functional partial dependence plot for predicted conditional quantiles versus the input scalar covariates in Figure \ref{fig: ap_income_PDP}. The abbreviation information are as follows, EDU: Education, ENV: Environment, PPL: Population, GDP: GDP Per Capita, CRM: Crime Rate, DBT: Diabetes, EMP: Unemployment. Since there is no ground truth in the real dataset, this figure only explains how each method interprets the functional dependence between the density output and scalar covariates in different ways.

\begin{figure*}[h]
	\begin{center}
		\includegraphics[width=\textwidth]{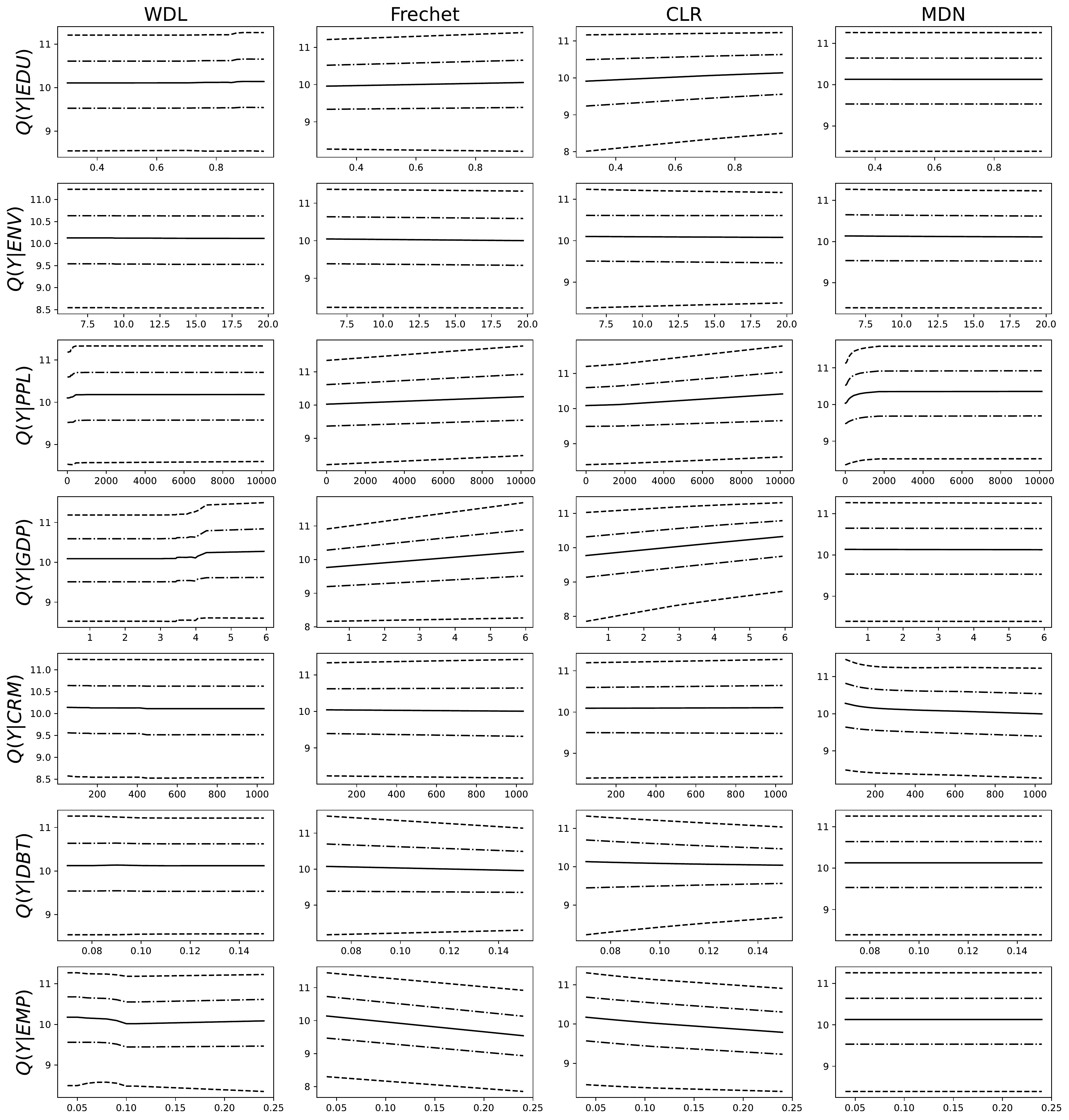}
		\caption{Functional partial dependence plot for predicted conditional quantiles versus the input scalar variables.}
		\label{fig: ap_income_PDP}
	\end{center}
\end{figure*}

\clearpage

\bibliographystyle{plain}
\bibliography{refs}

\end{document}